\documentclass[]{article}

\usepackage{times}
\usepackage{url,tabularx,array,enumerate,amsmath}
\usepackage{graphicx}
\usepackage[english]{babel}
\usepackage{natbib}
\usepackage{amsthm}
\usepackage{enumerate}
\usepackage{algorithm}
\usepackage{algpseudocode}
\usepackage{caption}
\usepackage{subcaption}
\usepackage[multiple]{footmisc}
\usepackage{thmtools,thm-restate}
\usepackage{verbatim}
\usepackage{xcolor}
\usepackage{ijcai17}

\newtheorem{mydef}{Definition}
\newtheorem{theorem}{Theorem}

\newtheorem{lemma}{Lemma}

\newtheorem{corollary}{Corollary}

\newtheorem{supptheorem}{Supplemental Theorem}

\newtheorem{elemma}{Conditional Edge Lemma}
\newcommand{\E}{{\rm I\kern-.3em E}}
\theoremstyle{definition}

\newcommand{\indep}{\rotatebox[origin=c]{90}{$\models$}}
\newcommand{\notindep}{\not\!\perp\!\!\!\perp}

\title{Identification and Model Testing in Linear Structural Equation Models using Auxiliary Variables}

\author{Bryant Chen\\
IBM Research\\
\url{bryant.chen@ibm.com}
\And
Daniel Kumor\\
Purdue University\\
\url{dkumor@purdue.edu}
\And
Elias Bareinboim\\
Purdue University\\
\url{eb@purdue.edu}}

\begin{document}
\maketitle









\begin{abstract}

We developed a novel approach to identification and model testing in linear structural equation models (SEMs) based on auxiliary variables (AVs), which generalizes a widely-used family of methods known as instrumental variables. The identification problem is concerned with the conditions under which causal parameters can be uniquely estimated from an observational, non-causal covariance matrix. In this paper, we provide an algorithm for the identification of causal parameters in linear structural models that subsumes previous state-of-the-art methods. In other words, our algorithm identifies strictly more coefficients and models than methods previously known in the literature. Our algorithm builds on a graph-theoretic characterization of conditional independence relations between auxiliary and model variables, which is developed in this paper. Further, we leverage this new characterization for allowing identification when limited experimental data or new substantive knowledge about the domain is available. Lastly, we develop a new procedure for model testing using AVs.
\end{abstract}


\section{Introduction}

The problem of estimating causal effects is one of the fundamental problems in the data-driven sciences.  In order to estimate a causal effect, the desired effect must be \emph{identified} or uniquely expressible in terms of the probability distribution over the available data.  Causal effects are identified by design in randomized control trials, but in many applications, such experiments are not possible.  When only observational data is available, determining whether a causal effect is identified requires modeling the underlying causal structure, which is generally done using \emph{structural equation models} (SEMs) (also called \emph{structural causal models}) \citep{pearl:09,bareinboim2016causal}.  

A structural equation model consists of a set of equations that describe the underlying data-generating process for a set of variables.  
While SEMs, in their most general, non-parametric form do not require any assumptions about the form of these functions, in many fields, including machine learning, psychology, and the social sciences, linear SEMs are used.  A linear SEM consists of a set of equations of the form, $X = \Lambda X+ U$, where $X = [x_1 , ... , x_n]^t$ is a vector containing the model variables, $\Lambda$ is a matrix containing the \emph{coefficients} of the model, and $\Lambda_{ij}$ represents the direct effect of $x_i$ on $x_j$, and $U = [u_1 , ..., u_n]^{t}$ is a vector of normally distributed error terms, which represents omitted or latent variables.\footnote{Instrumental and auxiliary variables can also be used when normality is not assumed, but to simplify the proofs in the paper, we will, as is commonly done by empirical researchers, assume normality.}  The matrix $\Lambda$ contains zeroes on the diagonal, and $\Lambda_{ij} = 0$ whenever $x_i$ is not a cause of $x_j$.  The covariance matrix of $X$ will be denoted by $\Sigma$ and the covariance matrix over the error terms, U, by $\Omega$.  In this paper, we will restrict our attention to \emph{semi-Markovian} models \citep{pearl:09}, models where the rows of $\Lambda$ can be arranged so that it is lower triangular, and the corresponding graph is acyclic.  

When modeling using SEMs, researchers typically specify the model by setting certain entries of $\Lambda$ and $\Omega$ to zero (i.e. exclusion and independence restrictions), while leaving the rest of the entries as free parameters to be estimated from data\footnote{There are a number of algorithms for discovering the model structure from data\citep{spirtes2000causation,shimizu:etal06,pearl:09,zhang:09,mooij:etal16}. However, it is only in very rare instances that these methods are able to uniquely determine the model structure. As a result, model specification generally utilizes knowledge about the domain under study.}. Restricting a particular entry $\Lambda_{ij}$ to zero reflects the assumption that $Y_i$ has no direct effect on $Y_j$. Similarly, restricting $\Omega_{ij}$ to zero reflects the assumption that there are no unobserved common causes of both $Y_i$ and $Y_j$. Once the parameters are estimated, causal effects (as well as counterfactual quantities) can be computed from the structural coefficients directly \citep{pearl:09,chen:pea14}. However, in order to be estimable from data, a parameter must first be identified. In some cases, the modeling assumptions are not strong enough, and there are multiple, often infinite, values for the parameter that are consistent with the observed data. As a result, two fundamental problems in SEMs are to identify and estimate the model parameters and to test the underlying assumptions that enable identification.


The problem of identification has been studied extensively by econometricians and social scientists \citep{fisher:66,bowden:tur84,bekker1994identification,rigdon1995necessary} and more recently by the AI and statistics communities using graphical methods \citep{spirtes:etal98,tian:07a,tian2009parameter,brito:pea02a,brito:pea02c,brito:pea06,bareinboim2016causal}.  To our knowledge, the most general, efficient algorithm for model identification is the g-HT algorithm given by \citet{chen:16} combined with ancestor decomposition \citep{drton:wei16}. This method generalizes the half-trek algorithm of \citet{foygel:12} and utilizes ancestor decomposition, which expands on an idea by \citet{tian:05} where the model is decomposed into simpler sub-models. Graphical methods have also been applied to the problem of testing the causal assumptions embedded in an SEM.  For example, d-separation \citep{pearl:09} and overidentification \citep{pearl:04,chen:etal14} provide the means to discover testable implications of the model, which can be used to test it against data. 

Despite decades of attention and work from diverse fields, the identification problem\footnote{To be precise, we are referring to the problem of identification almost everywhere \citep{brito:pea02b}, also called generic identification \citep{foygel:12}.} has still not been efficiently solved\footnote{An exhaustive procedure can be obtained using Gr\"obner bases methods \citep{foygel:12}. However, these methods are computationally intractable for anything but the smallest of graphs.}.  There are identifiable parameters and models that none of the above methods are able to identify.  Similarly, there are testable implications of SEMs that the above methods are unable to detect.  One promising avenue to aid in both tasks are \emph{auxiliary variables} \citep{chen:etal15}.  Each of the aforementioned methods for identification and model testing only utilizes restrictions on the entries of $\Lambda$ and $\Omega$ to zero.  Auxiliary variables can be used to incorporate knowledge of non-zero coefficient values into existing methods for identification and model testing.  These coefficient values could be obtained, for example, from a previously conducted randomized experiment, from substantive understanding of the domain, or even from another identification technique. The intuition behind auxiliary variables is simple: if the coefficient from variable $w$ to $z$, $\beta$, is known, then we would like to remove the direct effect of $w$ on $z$ by subtracting it from $z$. This removal eliminates confounding paths through $w$ and is performed by creating a variable $z^* = z-\beta w$, which is used as a proxy for $z$.  In many cases, $z^*$ allows the identification of parameters or testable implications using existing methods when $z$ could not.

\citet{chen:etal15} demonstrated how auxiliary variables could be utilized in simple instrumental sets (instrumental sets that do not utilize conditioning to block spurious paths) \citep{brito:pea02a,zander:etal15} and proved that any model identifiable using the g-HT algorithm is also identifiable using auxiliary simple instrumental sets. 

Since auxiliary variables allow knowledge of non-zero coefficient values to be incorporated into existing methods for identification, they are also directly applicable to the problem of z-identification \citep{bareinboim:pea12}, in which partial experimental data is available. Additionally, the cancellation of paths that results from adding an AV may result in conditional independence constraints between the AV and other variables that can be used to test the model. 



In this paper, we generalize the results of \citet{chen:etal15} and demonstrate how auxiliary variables can be utilized in generalized instrumental sets, which allow for conditioning to block spurious paths. We prove that, unlike auxiliary simple instrumental sets, this generalization \emph{strictly} subsumes the g-HT algorithm. Additionally, we introduce quasi-instrumental sets, which utilize auxiliary variables to identify coefficients when partial experimental data is available. Quasi-instrumental sets are incorporated into our identification algorithm, allowing it to better address the problem of z-identification. To our knowledge, this algorithm is the first systematic method for tackling z-identification in linear systems. We also demonstrate how auxiliary instrumental sets and quasi-instrumental sets can be used to derive over-identifying constraints, which can be used to test the model specification against data. Moreover, we prove that these overidentifying constraints subsume conditional independence constraints among auxiliary variables. Lastly, we discuss related work, showing how auxiliary IVs are able to unite a variety of disparate methods under a single framework.


\section{Preliminaries}
\label{sec:prelim}

The causal graph or path diagram of an SEM is a graph, $G=(V,D,B)$, where $V$ are nodes or vertices, $D$ directed edges, and $B$ bidirected edges.  The nodes represent model variables.  Directed eges encode the direction of causality, and for each coefficient $\Lambda_{ij}\neq 0$, an edge is drawn from $x_i$ to $x_j$.  Each directed edge, therefore, is associated with a coefficient in the SEM, which we will often refer to as its structural coefficient.  Additionally, when it is clear from context, we may abuse notation slightly and use coefficients and directed edges interchangeably.  The error terms, $u_i$, are not shown explicitly in the graph.  However, a bidirected edge between two nodes indicates that their corresponding error terms may be statistically dependent while the lack of a bidirected edge indicates that the error terms are independent.  

We will use standard graph terminology with $Pa(y)$ denoting the parents of $y$, $Anc(y)$ denoting the ancestors of $Y$, $De(y)$ denoting the descendants of $y$, and $Sib(y)$ denoting the siblings of $y$, the variables that are connected to $y$ via a bidirected edge.  $He(E)$ denotes the heads of a set of directed edges, $E$, while $Ta(E)$ denotes the tails.  Additionally, for a node $v$, the set of edges for which $He(E)=v$ is denoted $Inc(v)$.  Lastly, we will utilize d-separation \citep{pearl:09}.

We will use $\sigma (x, y|W)$ to denote the partial covariance between two random variables, $x$ and $y$, given a set of variables, $W$, and $\sigma (x,y|W)_G$ as the partial covariance between random variables $x$ and $y$ given $W$ implied by the graph $G$.  We will assume without loss of generality that the model variables have been standardized to mean 0 and variance 1.

\begin{mydef}
For a given unblocked (given the empty set) path, $\pi$, from $x$ to $y$, Left($\pi$) is the set of nodes, if any, that has a directed edge leaving it in the direction of $x$ in addition to $x$.  Right($\pi$) is the set of nodes, if any, that has a directed edge leaving it in the direction of $y$ in addition to $y$.  
\end{mydef}

For example, consider the path $\pi = x\leftarrow v^L_1 \leftarrow ... \leftarrow v^L_k \leftarrow v^T \rightarrow v^R_j \rightarrow ... \rightarrow v^R_1 \rightarrow y$.  In this case, Left($\pi$) $= \cup_{i=1}^{k} v^L_i \cup \{x, v^T\}$ and Right($\pi$) $= \cup_{i=1}^{j} v^R_i \cup \{y, v^T\}$.  $v^T$ is a member of both Right$(\pi)$ and Left($\pi)$.  

\begin{mydef}
A set of paths, ${\pi_1, ..., \pi_n}$, has \emph{no sided intersection} if for all $\pi_i , \pi_j \in \{\pi_1, ..., \pi_n\}$ such that $\pi_i \neq \pi_j$, Left$(\pi_i)\cap$Left$(\pi_j)$=Right$(\pi_i)\cap$Right$(\pi_j)=\emptyset$.
\end{mydef}

Wright's rules \citep{wright:21} allow us to equate the model-implied covariance, $\sigma (x,y)_M$, between any pair of variables, $x$ and $y$, to the sum of products of parameters along unblocked paths between $x$ and $y$.\footnote{Wright's rules characterize the relationship between the covariance matrix and model parameters.  Therefore, any question about identification using the covariance matrix can be decided by studying the solutions for this system of equations.  However, since these equations are polynomials and not linear, it can be very difficult to analyze identification of models using Wright's rules.}  Let $\Pi= \{\pi_1 , \pi_2 ,... , \pi_k\}$ denote the unblocked paths between $x$ and $y$, and let $p_i$ be the product of structural coefficients along path $\pi_i$.  Then the covariance between variables $x$ and $y$ is $\sum_i p_i$.

\begin{figure*}
\centering
\begin{subfigure}[t]{.22\textwidth}
\caption{}
\label{fig:EC-IV}
\includegraphics[width=\textwidth]{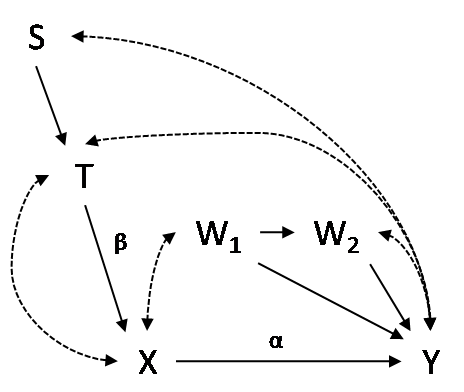}
\end{subfigure}
\begin{subfigure}[t]{.225\textwidth}
\caption{}
\label{fig:EC-IV2}
\includegraphics[width=\textwidth]{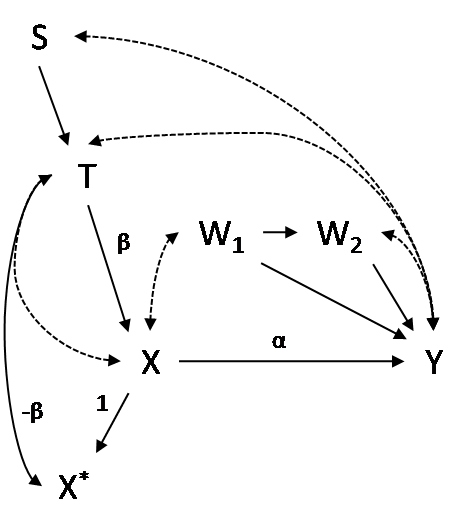}
\end{subfigure}
\begin{subfigure}[t]{.225\textwidth}
\caption{}
\label{fig:nodesc}
\includegraphics[width=\textwidth]{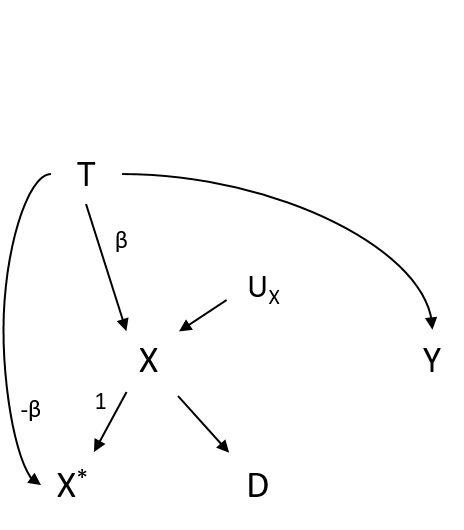}
\end{subfigure}
\caption{(a) $\alpha$ is not identified using IVs (b) $\alpha$ is identified using $x^*$ as an auxiliary IV given $w_1$ (c) conditioning on descendants of $x$ induces correlation between $x^*$ and $y$}
\end{figure*}

Lastly, we define auxiliary variables and the augmented graph.  

\begin{mydef}[Auxiliary Variable]
\label{def:auxvar}
Given a linear SEM with graph $G$ and a set of edges $E$ whose coefficient values are known, an \emph{auxiliary variable} is a variable, $z^* = z-\sum_i e_it_i$, where $\{e_1, ..., e_k\}\subseteq E \cap Inc(z)$ and $t_i = Ta(e_i)$ for all $i\in \{1,..., k\}$.   
\end{mydef}

If not otherwise specified, $z^*$ refers to the auxiliary variable, $z-c_1 t_1 -... -c_l t_l$, where $\{c_1, ..., c_l\}$ are the coefficients of $E\cap Inc(z)$ and $E$ is the set of directed edges whose coefficient values are known.  In other words, $z^*$ is the auxiliary variable for $z$ where as many known coefficients are subtracted out as possible.  \citet{chen:etal15} demonstrated that the covariance between any auxiliary variables and model variables can be computed using Wright's rules on the \emph{augmented graph}, defined below.  

\begin{mydef}
\label{def:aug}
\citep{chen:etal15} Let $M$ be a linear SEM with graph $G$ and a set of directed edges $E$ such that their coefficient values are known. The \emph{$E$-augmented} model, $M^{E+}$, includes all variables and structural equations of $M$ in addition to new auxiliary variables, $y^*_1 , ... y^*_k$, one for each variable in $He(E)=\{y_1 , ... , y_k\}$  such that the structural equation for $y^*_i$ is $y^*_i = y_i - \Lambda_{X_i y_i} T^t_i$, where $X_i=Ta(E)\cap Pa(y_i)$, for all $i\in \{1 ,... , k\}$.  The corresponding \emph{augmented graph} is denoted $G^{E+}$.  

\end{mydef}

For example, consider Figure \ref{fig:EC-IV}.  If the value of $\beta$ is known, we can generate an auxiliary variable $x^* = x - \beta t$.  The $\beta$-augmented graph $G^{\beta+}$ is depicted in Figure \ref{fig:EC-IV2}.  In some cases, $x^*$ allows the identification of coefficients and testable implications using existing methods when $x$ could not, due to the fact that the back-door paths from $x$ to $y$ that go through $\beta$ cancel with the back-door paths from $x^*$ to $y$ that go through $-\beta$.  This can be seen by expressing the covariance of $x^*$ and $y$ in terms of the model parameters using Wright's rules.  

\section{Auxiliary and Quasi-Instrumental Sets}
\label{sec:q-IS}

Two, perhaps the most common, methods for estimating causal effects are OLS regression and two-stage least-squares (2SLS) regression.  Both of these methods assume that the underlying causal relationships between variables are linear, in addition to other causal assumptions that guarantee identification.  The \emph{single-door criterion} \citep{pearl:09} graphically characterizes when the assumptions sufficient to estimate a causal effect using regression are satisfied in a linear SEM.  Similarly, \citet{brito:pea02a} gave a graphical characterization for when a variable $z$ qualifies as an IV so that 2SLS regression provides a consistent estimate of the causal effect.  In this section, we give a graphical criterion for when AVs can be utilized in generalized instrumental sets, which extends both the single-door criterion and IVs. Additionally, we introduce quasi-instrumental sets, which utilize AVs to better address the problem of z-identification.

First, we give a simple graphical criterion for when an AV would be conditionally independent of another variable, which will allow us to incorporate AVs into instrumental sets, as well as other identification and model testing methods that require the ability to detect conditional independence in the graph.

\begin{theorem}
\label{thm:sep}
Given a linear SEM with graph $G$, where $E\subseteq Inc(z)$ is a set of edges whose coefficient values are known, if $W\cup \{y\}$ does not contain descendants of $z$ and $G_{E-}$ represents the graph $G$ with the edges for $E$ removed, then $(z^*\indep y | W)_{G^{E+}}$ if and only if $(z\indep y | W)_{G_{E-}}$.\footnote{The theorem disallows descendants of the generating variable in the conditioning set. At first glance, this may appear to limit the ability to block biasing paths among AVs. However, we conjecture that if $z$ cannot be separated from $y$ in $G$, then $z^*$ will almost surely not be independent of $y$ given $W$, if $W$ contains descendants of $z$. To illustrate, consider the example shown in Figure \ref{fig:nodesc}. $x^* = x - \beta t$ is independent of $y$, as can be verified using Wright's rules, but $x^*$ is not independent of $y$ given $d$! An intuitive explanation for this surprising result is that conditioning on $d$, a descendant of $x$, in Figure \ref{fig:nodesc} induces correlation between the error term of $x$ and $t$, since $x$ acts as a ``virtual collider''.  As a result, we have a ``virtual path'' from $x^*$ to $y$, $x^*\leftarrow x \leftarrow u_x \leftrightarrow t \rightarrow y$. See \citet[p. 339]{pearl:09} for a detailed discussion of virtual colliders.}
\end{theorem}
\begin{proof}
Proofs for all theorems and lemmas can be found in the Appendix.
\end{proof}



Next, we demonstrate how AVs can be incorporated into generalized instrumental sets, defined below.

\begin{theorem}
\label{thm:IS}
\citep{brito:pea02a} Given a linear model with graph $G$, the coefficients for a set of edges $E = \{(x_1, y), ..., (x_k, y)\}$ are identified if there exists triplets $(z_1, W_1, p_1), ..., (z_k, W_k, p_k)$ such that for $i=1, ..., k$,
\begin{enumerate}[(i)]
\item $(z_i \indep y |W_i)_{G_{E-}}$, where $W$ does not contain any descendants of $y$ and $G_{E-}$ is the graph obtained by deleting the edges, $E$ from $G$,
\item $p_i$ is a path between $z_i$ and $x_i$ that is not blocked by $W_i$, and
\item the set of paths, $\{p_1, ..., p_k\}$ has no sided intersection.\footnote{\citet{brito:pea02a} provided an alternative statement of condition (iii). A proof that the two statement are, in fact, equivalent is given in the Appendix.}
\end{enumerate}
If the above conditions are satisfied, we say that $Z$ is a \emph{generalized instrumental set} for $E$ or simply an \emph{instrumental set} for $E$.\footnote{Note that when $k=1$, $z_1$ is a conditional IV for $(x_1, y)$.  Further, if $z_1 = x_1$, then $x_1$ satisfies the single-door criterion for $(x_1 , y)$. The converse is not true, however (see appendix \ref{sec:cIVvsgIS}).}
\end{theorem}
In some cases, a variable $z$ may not satisfy condition (i) above but an auxiliary variable $z^*$ does.  For example, in Figure \ref{fig:EC-IV}, we cannot identify $\alpha$ using Theorem \ref{thm:IS}.  Blocking the path $x\leftarrow t \leftrightarrow y$ by conditioning on $t$ opens the path, $x\leftrightarrow t \leftrightarrow y$.  Moreover, we cannot use $t$ or $s$ in an instrumental set due to the edges $t\leftrightarrow y$ and $s\leftrightarrow y$.  However, $s$ is an IV for $\beta$, allowing us to generate an AV, $x^* = x-\beta\cdot t_1$, as in Figure \ref{fig:EC-IV2}.  Now, $\alpha$ can be identified using $x^*$ as an auxiliary instrument given $w_1$.

Theorem \ref{thm:sep} tells us when (i) of Theorem \ref{thm:IS} can be satisfied using an AV, $z_i^*$. We simply check whether $z_i$ can be separated from $y$ in $G_{E\cup E_z-}$, where $E_z \subseteq Inc(z_i)$ is the set of $z_i$'s edges whose coefficient values are known. When an instrumental set includes AVs, we call the set an \emph{auxiliary instrumental set} or \emph{auxiliary IV set} for short. 

Figure \ref{fig:EC-IV} also demonstrates the importance of extending the simple auxiliary instrumental sets introduced by \citet{chen:etal15} to allow for conditioning. $\alpha$ can only be identified if we block the paths $x\leftrightarrow w_1 \rightarrow y$ and $x\leftrightarrow w_1 \rightarrow w_2 \rightarrow y$ by conditioning on $w_1$.

\begin{figure}
\centering
\begin{subfigure}[t]{.17\textwidth}
\caption{}
\label{fig:auxY}
\includegraphics[width=\textwidth]{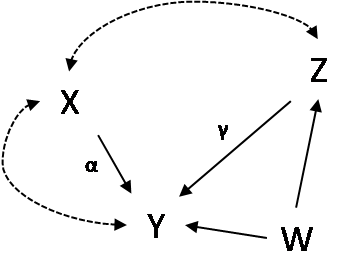}
\end{subfigure}
\begin{subfigure}[t]{.192\textwidth}
\caption{}
\label{fig:auxY2}
\includegraphics[width=\textwidth]{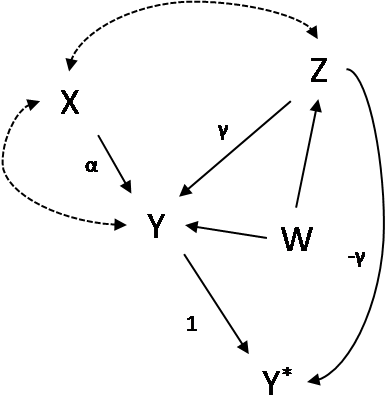}
\end{subfigure}
\caption{(a) $\alpha$ is not identified using IVs (b) $\alpha$ is identified using $Z$ as a quasi-IV after adding auxiliary variable $Y^*$}
\end{figure}

When knowledge of coefficient values are known a priori, it may be helpful to generate an AV from the outcome variable $y$. For example, in Figure \ref{fig:auxY}, $\alpha$ cannot be identified. However, suppose that it is possible to run a surrogate experiment and randomize $z$. This experiment would allow us to estimate $\gamma$ and generate the AV, $Y^* = Y - \gamma Z$. Now, $z$ is not technically an instrument for $\alpha$, but it can be shown that $\alpha = \frac{r_{Y*Z.W}}{r_{XZ}}$. \citet{chen:etal15} called such variables \emph{quasi-instrumental variables} or \emph{quasi-IVs} for short.

Interestingly, while quasi-IVs are valuable for the problem of z-identification, they do no better than instrumental sets when applied to the standard identification problem, where no external knowledge of coefficient values is available. For example, consider again Figure \ref{fig:auxY}. In order to use $z$ as a quasi-IV for $\alpha$, we would first have to identify $\gamma$ using an IV. If such a variable existed, say $z'$, then we could have simply identified $\{\alpha, \gamma\}$ using the IV set $\{z,z'\}$.


\begin{figure*}
\centering
\begin{subfigure}[t]{.33\textwidth}
\caption{}
\label{fig:iter-aux1}
\includegraphics[width=\textwidth]{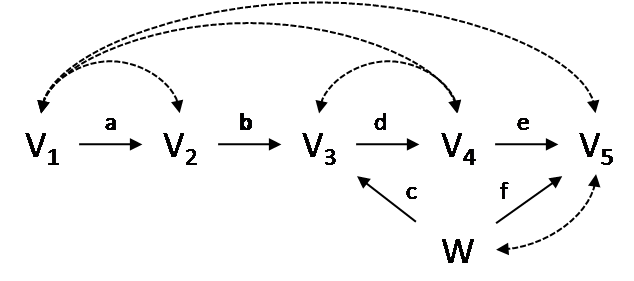}
\end{subfigure}
\begin{subfigure}[t]{.33\textwidth}
\caption{}
\label{fig:iter-aux2}
\includegraphics[width=\textwidth]{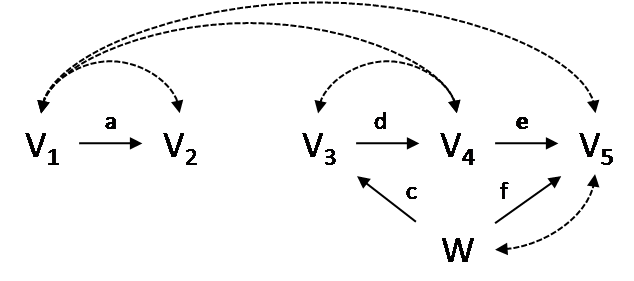}
\end{subfigure}
\begin{subfigure}[t]{.33\textwidth}
\caption{}
\label{fig:iter-aux3}
\includegraphics[width=\textwidth]{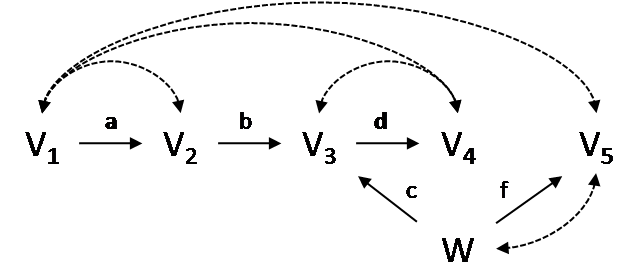}
\end{subfigure}
\caption{(a) $b$ is identified using either $v_2$ or $v_1$ as an instrument and $c$ is identified using $w$ as an instrument (b) $e$ is identified using $v_3^*$ as an auxiliary instrument given (c) $a$ and $d$ are identified using $v_5^*$ as an auxiliary instrument}
\end{figure*}

Next, we formally define \emph{quasi-instrumental sets} or \emph{quasi-IV sets} for short. Note that auxiliary IV sets are also quasi-IV sets.

\begin{restatable}{mydef}{aqiv}
\label{def:aqiv}
Given a linear SEM with graph G, a set of edges $E_K$ whose coefficient values are known, and a set of structural coefficients $\alpha = \{\alpha_1,\alpha_2,...,\alpha_k\}$, the set $Z=\{z_1,...,z_k\}$ is a quasi-instrumental set if there exist triples $(z_1,W_1,p_1),...,(z_k,W_j,p_k)$ such that:

\begin{enumerate}[(i)]
\item For $i=1,...,k$, either:
\begin{enumerate}
\item the elements of $W_i$ are non-descendants of $y$, and $(z_i\indep y|W_i)_{G_{E\cup E_y}}$ where $E_y = E_K \cap Inc(y)$.
\item the elements of $W_i$ are non-descendants of $z_i$ and $y$, and $(z_i\indep y|W_i)_{G_{E\cup E_{zy}}}$ where $E_{zy}= E_K \cap (Inc(z)\cup Inc(y))$.
\end{enumerate}
\item for $i=1,...,k$, $p_i$ is a path between $z_i$ and $x_i$ that is not blocked by $W_i$, where $x_i=Ta(\alpha_i)$, and
\item the set of paths $\{p_1,...,p_k\}$ has no sided intersection
\end{enumerate}
\end{restatable}

\begin{restatable}{theorem}{aqivid}
\label{thm:aqivid}
If $Z^*$ is a quasi-instrumental set for $E$, then $E$ is identifiable.
\end{restatable}

Lastly, the following corollary provides a simple graphical condition for when a single variable or AV qualifies as a quasi-IV.

\begin{corollary}
\label{cor:aux}
Given a linear SEM with graph $G$, $z^*$ is a quasi-IV for $\alpha$ given $W$ if $W$ does not contain any descendants of $z$, and $z$ is an IV for $\alpha$ given $W$ in $G_{E_z \cup E_y-}$, where $E_z\subseteq Inc(z)$ and $E_y\subseteq Inc(y)$ are sets of edges whose coefficient values are known.

\end{corollary}

Auxiliary and quasi-IV sets enable a bootstrapping procedure whereby complex models can be identified by iteratively identifying coefficients and using them to generate new auxiliary variables.  For example, consider Figure \ref{fig:iter-aux1}.  First, we are able to identify $b$ and $c$ using IVs, but no other coefficients.  Once $b$ is identified, Corollary \ref{cor:aux} tells us that $e$ is identified using $v_3^*$ since $v_3$ is an IV for $e$ when the edge for $b$ is removed (see Figure \ref{fig:iter-aux2}).  Now, the identification of $e$ allows us to identify $a$ and $d$ using $v_5^*$, since $v_5$ is an IV for $a$ and $d$ when the edge for $e$ is removed (see Figure \ref{fig:iter-aux3}). This general strategy is the basis for our identification, z-identification, and model testing algorithm, described next.

\section{Identification and z-Identification Algorithm}
\label{sec:alg}

In this section, we construct an identification algorithm that operationalizes the bootstrapping approach described in Section \ref{sec:q-IS}.  First, we describe how to algorithmically find a quasi-instrumental set for a set of coefficients $E$, given a set of known coefficients, $\mathrm{IDEdges}$.  

The problem of finding generalized instrumental sets was addressed by \citet{zander:etal16}.  They provided an algorithm, $\mathrm{TestGeneralIVs}$, that determines whether a given set $Z$ is a generalized instrumental set for a set of edges, $E$, that runs in polynomial time if we bound the size of the coefficient set to be identified.  More specifically, their algorithm has a running time of $O((k!)^2 n^k)$, where $n$ is the number of variables in the graph and $k=|E|$.\footnote{\citet{zander:etal16} also give an algorithm that tests whether $Z$ is a \emph{simple conditional instrumental sets} in $O(nm)$ time.  A simple conditional instrumental set is a generalized instrumental set where $W_1 = W_2 = ... = W_k$}  

Our method, $\mathrm{TestQIS}$, given in the Appendix, generalizes $\mathrm{TestGeneralIVs}$, for quasi-IV sets. $\mathrm{FindQIS}$, also given in the Appendix, searches for a quasi-IV set by checking all subsets of $Z\subseteq (Anc(z_i)\cup Anc(y))$ using $\mathrm{TestQIS}$. It returns a quasi-IV set, as well as its conditioning sets, if one exists.

In some cases an instrumental set may not exist for $C$, but one exists for $C^{'}$, where $C\subset C^{'}$.  Conversely, there may not be an instrumental set for $C^{'}$, but there is one for $C\subset C^{'}$.  As a result, we may have to check all possible subsets of a variable's coefficients in order to determine whether a given subset is identifiable using auxiliary instrumental sets.

The ID algorithm, called $qID$ utilizes $\mathrm{FindQIS}$ to identify as many coefficients as possible in a given model with graph $G$. It iterates through each variable and attempts to identify its parents using $\mathrm{FindQIS}$. If it is unable to identify the parents, it then attempts to identify subsets of the parents. After the algorithm has attempted to identify each subset for each variable, it again attempts to identify each unidentified subset, since each newly identified coefficient may enable the identification of previously unidentifiable coefficients. This process is repeated until all coefficients have been identified or no new coefficients have been identified in the last iteration. The algorithm is polynomial if the degree of each node in the graph is bounded.

Our algorithm identifies the model depicted in Figure \ref{fig:model_subsume} in the following way.  First, let us assume that the variables are arbitrarily ordered, so the sets of their incident edges can be: $(\{a\}, \{b,c,f\}, \{d\}, \{e\})$.  Now, the first edge to be identified would be $a$ using $w_1$ as an IV.  There is no auxiliary IV set for $\{b,c,f\}$, and we would attempt to find one for its subsets.  We find that $\{b\}$ is identified using $\{x\}$ as an IV set with conditioning set $\{w_1\}$.  Now, $\{d\}$ is identified using $y^*=y-bx$, and $e$ is identified using $t^*_2$.  In the second iteration, we return to $\{b, c, f\}$ and find that it is now identified using the auxiliary IV set, $\{x, w_1, t^*_3\}$.

In contrast, Figure \ref{fig:model_subsume} is not identified using simple instrumental sets and auxiliary variables.  We cannot identify $b$ without conditioning on $w_1$, which means that the only coefficients identified using auxiliary simple instrumental sets is $a$.  Since \citet{chen:etal15} showed that any coefficient identified using the generalized half-trek criterion (g-HTC) can be identified using auxiliary variables and simple instrumental sets, we know that $qID$ is able to identify coefficients and models that the g-HT algorithm is not.

\begin{algorithm}[H]
\caption{qID($G,\Sigma, \mathrm{IDEdges}$)}
\label{alg:ID}
\begin{algorithmic}
\State {\bfseries Initialize:} $\mathrm{EdgeSets}\leftarrow$ $Inc(v)$ for $v\in V$
\Repeat
	\ForAll{$ES$ in $\mathrm{EdgeSets}$ such that \\ \indent \indent$ES\not\subseteq \mathrm{IDEdges}$}
		\State $y\leftarrow He(ES)$
			\ForAll{$E\subseteq ES$ such that $E\not\subseteq \mathrm{IDEdges}$}
			\State $(Z,W)\leftarrow \mathrm{FindQIS}(G,ES,\mathrm{IDEdges})$
				\If{$(Z, W)\neq \perp$}					
					\State Identify $E$ using $Z^*$ as an auxiliary \\ \indent\indent\indent\indent instrumental set in $G^{(\mathrm{IDEdges}\cap Inc(Z))+}$
					\State $\mathrm{IDEdges}\leftarrow \mathrm{IDEdges} \cup E$
				\EndIf
			\EndFor
	\EndFor
\Until{All coefficients have been identified or no coefficients have been identified in the last iteration}\\
\end{algorithmic}
\end{algorithm}

Moreover, $qID$ will identify any coefficients that are identifiable using auxiliary variables and simple instrumental sets, giving us the following theorem.

\begin{theorem}
\label{thm:g-HT-subsumed}
Given an arbitrary linear causal model, if a set of coefficients is identifiable using the g-HT algorithm, then it is identifiable using $qID$.  Additionally, there are models that are not identified using the g-HT algorithm, but identified using $qID$.
\end{theorem}

\section{Deriving Testable Implications using AVs}

Theorem \ref{thm:sep} also enables us to derive new vanishing partial correlation constraints that can be used to test the model. For example, in Figure \ref{fig:overID}, $\alpha$ can be identified using $z_1$ as an instrument. Once $\alpha$ is identified, we can generate the AV $y^*=y-\alpha x=y-\frac{\sigma(y, z_1)}{\sigma(x, z_1)}x$, and Theorem \ref{thm:sep} tells us that the correlation of $z_2$ and $y^*$ should vanish. As a result, we can test the model specification by verifying that this constraint holds in the data.

\begin{figure}
\centering
\begin{subfigure}[b]{.185\textwidth}
\caption{}
\label{fig:overID}
\includegraphics[width=\textwidth]{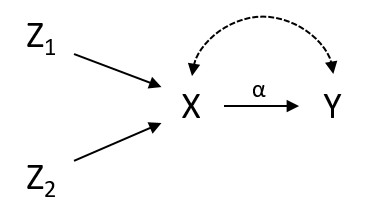}
\end{subfigure}
\begin{subfigure}[b]{.208\textwidth}
\caption{}
\label{fig:model_subsume}
\includegraphics[width=\textwidth]{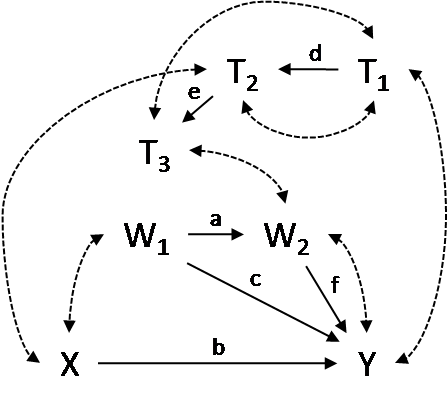}
\end{subfigure}
\caption{(a) $\sigma(z_2, y^*)=0$, where $y^*=y-\frac{\sigma(y,z_1)}{\sigma(x,z_1)}$, and, equivalently, $\alpha$ is overidentified using $z_1$ and $z_2$ as IVs (b) the model is identified using auxiliary instrumental sets, but not the g-HT algorithm}
\end{figure}

Theorem \ref{thm:sep} also tells us that the correlation between $z_1$ and $y^*$ should also vanish. However, upon closer inspection, we find that this implication does not actually constrain the covariance matrix:
\begin{align*}
\sigma(z_1, y^*) =& \sigma(z_1, y-\alpha x) \\
=& \sigma(z_1, y) - \frac{\sigma(y,z_1)}{\sigma(x, z_1)}\sigma(z_1, x) =0.
\end{align*}

In other words, our ``testable implication" that $\sigma(z_1, y^*)=0$ is equivalent to stating $\sigma(z_1, y)-\sigma(z_1,y)=0$--a tautology! In contrast, 
\begin{equation*}
\sigma(z_2, y^*) = \sigma(z_2, y) - \frac{\sigma(z_1, y)}{\sigma(x, z_1)}\sigma(z_2, x) =0
\end{equation*}
does provide a true testable implication.

\citet{shpitser2009testing} noticed a similar phenomenon when deriving dormant independences in non-parametric models, and their explanation applies to conditional independence constraints among AVs as well. The idea is the following: When the model implies that two variables are conditionally independent, it relies on the modeled assumption that there is no edge between those variables. As a result, verifying that the constraint holds in data represents a test that this assumption is valid. However, unlike conditional independence constraints between model variables, conditional independence constraints among AVs rely upon the absence of certain edges in order to identify the coefficients necessary to generate the AV. The key point is that this identification cannot rely on the same lack of edge whose existence we are trying to test!

In the above example, we identified $\alpha$ using $z_1$ as an IV. $\sigma(z_2, y^*)=0$ follows from the lack of edge between $z_2$ and $y$. However, even if this edge did exist, $z^*$ still equals $z-\frac{\sigma(y,z_1)}{\sigma(x,z_1)}x$. In contrast, $\sigma(z_1, y^*)=0$ follows from the lack of edge between $z_1$ and $y$. The existence of this edge would disallow $z_1$ as an instrument and $z^* = z-\alpha x \neq z-\frac{\sigma(y,z_1)}{\sigma(x,z_1)}x$.

Another way to derive the constraint $\sigma(z_2, y^*)=0$ is via overidentification. $\alpha$ can be identified using either $z_1$ or $z_2$ and equating the corresponding expressions yields the constraint $\frac{\sigma(y,z_1)}{\sigma(x,z_1)}=\frac{\sigma(y,z_2)}{\sigma(x,z_2)}$, which is clearly equivalent to the previous constraint $\sigma(z_2, y^*)=0$. In fact, we show (Theorem \ref{thm:aCI-subsumed}) that whenever a variable $z$ cannot be separated from another variable $y$, but $z^*$ can be, the resulting AV conditional independence, if it is non-vacuous, is equivalent to an overidentifying constraint that can be derived using quasi-IVs. As a result, all non-vacuous AV conditional independences are captured by overidentifying constraints derived using quasi-IVs!

First, we give a sufficient condition for when a set of edges $\alpha$ is overidentified. 

\begin{restatable}{theorem}{overid}
\label{thm:overid}
Let $Z$ be a quasi-IV set for structural coefficients $\alpha=\{\alpha_1,...,\alpha_k\}$ and $E$ be a set of known edges. If there exists a node $s$ satisfying the conditions listed below, then $\alpha$ is overidentified and we obtain the constraint .
\begin{enumerate}[(i)]
\item $s\notin Z$
\item There exists an unblocked path between $s$ and $y$ including an edge in $\alpha$
\item There exists a conditioning set $W$ that does not block the path $p$, such that either:
\begin{enumerate}
\item the elements of $W$ are non-descendants of $y$, and $(s\indep y|W)_{G_{\alpha\cup E_y}-}$, where $E_y = E \cap Inc(y))$
\item the elements of $W$ are non-descendants of $s$ and $y$, and $(s\indep y|W)_{G_{\alpha \cup E_{s}\cup E_y-}}$ where $E_s = E \cap Inc(s)$.
\end{enumerate}
\end{enumerate}
\end{restatable}


The above theorem can be used to derive an overidentifying constraint for every variable that satisfies (i)-(iii) above. It can also be applied when $\alpha$ is known a priori, yielding a \emph{z-overidentifying constraint}. In this case, $Z=\emptyset$ would be a quasi-IV set that trivially identifies $\alpha$. 

The following theorem states that non-vacuous AV conditional independence constraints are subsumed by quasi-IV overidentifying and z-overidentifying constraints.

\begin{theorem}
\label{thm:aCI-subsumed}
Let $z^*=z-e_1 t_1 - ... - e_k t_k$ and suppose there does not exist $W$ such that $(z\indep y|W)_G$. There exists $W$ such that $W\cap De(z) = \emptyset$ and $(z^*\indep y|W)$ is non-vacuous if and only if $y$ satisfies the conditions of Theorem \ref{thm:overid} for $E=\{e_1 ,... , e_k\}$.
\end{theorem}

The above theorem also applies when $y$ is an AV, called $y^*$. In this case, we simply replace $(z\indep y|W)_G$ with $(z\indep y^*|W)_{G^{E_y+}}$, where $E_y \subseteq Inc(y)$ is a set of edges whose coefficient values are known. 

Algorithm \ref{alg:overid} uses quasi-IV sets to output overidentifiying constraints in a graph given an optional set of identified edges. It uses $\mathrm{isEIV}$, which is a slightly modified version of $\mathrm{FindQIS}$ that tests whether $w$ fits the conditions of Theorem \ref{thm:aCI-subsumed}. Details of $\mathrm{isEIV}$ can be found in the Appendix.

\begin{algorithm}[H]
\caption{Finds overidentifying constraints for $G$ }
\label{alg:overid}
\begin{algorithmic}
\Function{constraintFinder}{G,$\Sigma$,IDEdges}
    \ForAll{$ES\in$ Edge Sets of $G$}
      \State $(Z,W) \gets$ \Call{FindQIS}{ES,G,IDEdges}
      \If{$(Z,W)\ne \bot$}
        \ForAll{$w \in V\setminus Z\cup \{He(ES)\}$}
          \If{\Call{IsEIV}{w,ES,G,IDEdges}}
            \State Add constraint $a_wA^{-1}b=b_w$
          \EndIf
        \EndFor
      \EndIf
    \EndFor
\EndFunction
\end{algorithmic}
\end{algorithm}

\section{Discussion and Related Work}
\label{sec:rw}
In this section, we discuss how (single-variable) auxiliary IVs encompass a number of previous identification methods developed in economics \citep{hausman:83}, computer science \citep{chan:kuroki10}, and epidemiology \citep{shardell:15}. 

\citet{hausman:83} showed that if the equation for a given variable, $z = \beta_1 p_1 + ... + \beta_k p_k + u_z$, is identified, then the error term $u_z$ can be estimated and used as an instrument for other coefficients. In this case, the auxiliary variable $z^* = z - \beta_1 p_1 - ... - \beta_k p_k$ is equal to the error term $u_z$. As a result, whenever the error term is estimable and can be used as an IV, we can also generate an auxiliary instrument.  However, there are times when only some of the coefficients in an equation are identifiable, and as a result, the error term cannot be used as an instrument, but we can nevertheless generate an auxiliary instrument. As a result, auxiliary IVs strictly subsume error term IVs.

\citet{chan:kuroki10} gave sufficient conditions for when a descendant of $x$ and a descendant of $y$ could be used in analogous manner to IVs to identify the effect of $x$ on $y$. In the context of AVs, this method is equivalent to generating an auxiliary instrument from the descendant by subtracting the total effect of $x$ on the descendant or the total effect of $y$ on the descendant (depending on whether the variable is a descendant of $x$ or $y$). In this paper, we generated AVs by subtracting out direct effects, but clearly the work can be extended to subtracting out total effects. The benefit of AVs over these descendant IVs is that they can be generated from a variety of variables, not just descendants of $x$ and $y$. Additionally, descendants of $x$ or $y$ can generate AVs from other total or direct effects, not just the effect of $x$ or $y$ on the descendant. 

The notion of ``subtracting out a direct effect'' in order to turn a variable into an instrument was also noted by \citet{shardell:15} when attemping to identify the total effect of $x$ on $y$. It was noticed that in certain cases, the violation of the independence restriction of a potential instrument $z$ (i.e. $z$ is not independent of the error term of $y$) could be remedied by identifying, using ordinary least squares regression, and then subtracting out the necessary direct effects on $y$.  AVs generalize and operationalize this notion so that it can be used on arbitrary sets of known coefficient values and be utilized in conjunction with existing graphical methods for identification and enumeration of testable implications. 

Additionally, as we have alluded to earlier, the highly algebraic, state-of-the-art g-HTC can also be understood in terms of auxiliary instruments. Identification using the g-HTC is equivalent to identification using auxiliary simple instrumental sets.

In summary, auxiliary instruments are not only the basis for the most general identification algorithm yet devised, but they also unify disparate identification methods under a single framework. Moreover, AVs are directly applicable to the tasks of z-identification and model testing. Finally, they can, in principle, enhance any method for identification, model testing, or other tasks that relies on graphical separation. 

\section{Conclusion}

In this paper, we graphically characterized conditional independence among AVs, allowing us to demonstrate how they can help generalized instrumental sets in the problem of identification. We provided an algorithm that identifies more models than the g-HT algorithm, subsuming the state-of-the-art for identification in linear models. Additionally, we introduced quasi-IV sets, and constructed an algorithm that utilizes them to attack the problem of z-identification. Finally, we proved that AV conditional independences are subsumed by overidentifying constraints and gave an algorithm for deriving overidentifying constraints.

\section{Acknowledgments}

We would like to thank Judea Pearl, Mathias Drton, Thomas Richardson, and Luca Weihs for helpful discussions. This research was supported in parts by grants from NSF \#IIS-1302448, \#IIS-1527490, ONR \#N00014-13-1-0153, \#N00014-13-1-0153, and DARPA ARO W911NF-16-1-0579.

\appendix
\section{Proof That AVs Work}
\setcounter{theorem}{0}
\begin{restatable}{theorem}{thmOne}
\label{thm:theoremOne}
Given a linear SEM with graph $G$, where $E\subseteq Inc(z)$ is a set of edges whose coefficient values are known, if $W\cup \{y\}$ does not contain descendants of $z$, then $(z^*\indep y | W)_{G^{E+}}$ if and only if $(z\indep y | W)_{G_{E-}}$.
Furthermore, $\sigma_{z^*y\cdot W}^{G^{E+}} = \sigma_{zy\cdot W}^{G_{E-}}$.
\end{restatable}

\subsection{Notation}

The statement $(z\indep y | W)_{G_{E-}}$ is equivalent to saying: $\rho_{zy.W}=0$ in the graph with incoming
known edges removed.
Similarly, $(z^*\indep y | W)_{G^{E+}}$ is saying that $\rho_{z^*y.W}=0$ in the graph with added auxiliary variable.

Let $\Sigma$ be the covariance matrix containing covariances between z,y, and all elements of $W$.

Let $\Sigma^*$ be the equivalent matrix with $z$ replaced by $z^*$.

Finally, we will use the notation $\Sigma_{zW,yW}$ to represent the matrix with only the rows
corresponding to z and elements of $W$, and columns of $y$ and elements of $W$. That is,
the mentioned matrix has the $y$ row removed, and the $z$ column removed.

We will use the determinant formula for partial covariance.

$$
\sigma_{zy.W} = \frac{\det \Sigma_{zW,yW}}{\det \Sigma_{WW}}
$$

By the Gessel-Viennot-Lindstrom lemma as applied to mixed graphs (see t-separation paper),
we know that $\det \Sigma_{WW}\ne 0$, since there exist paths of length 0 from each $w\in W$
to itself that don't intersect.

Similarly, we have 

$$
\sigma_{z^*y.W} = \frac{\det \Sigma^*_{z^*W,yW}}{\det \Sigma^*_{WW}}
$$

Notice that $\Sigma^*_{WW}$ is just the covariance matrix of the weights in the unmodified graph,
meaning in the graph where we neither added the auxiliary variable nor deleted edges. This is because
none of the paths of the covariances go through the auxiliary variable, as it is a collider (remember that
the covariance matrix contains only unconditioned covariances). Same as above, we conclude that
this determinant is non-zero.

Therefore, the theorem's statement is effectively saying:

$$
\det \Sigma_{zW,yW} = 0\ \ \ \  \text{ iff }\ \ \ \ \det \Sigma^*_{z^*W,yW} = 0
$$

except when there are descendants of $z$ in $W$.

For clarity, the following notation will be used in the rest of this document:
\begin{itemize}
\item $\delta_{ab}$ represents all the \textbf{directed} paths from $a$ to $b$
\item $\gamma_{ab}=\sigma_{ab} - \delta_{ab}$ in a normalized model, meaning that $\gamma$
contains all up-paths, or all paths that start from an edge incoming to $a$, including paths starting with bidirected edges.
\end{itemize}

We will also be using $\gamma_{zy}^{(e)}$ as all back paths from $z$ to $y$ taking the AV edges $e$, 
and all back paths that do not take the AV edges as $\gamma^{(-e)}_{zy}$. Unless explicitly specified, these paths are assumed to be in graph $G^{E+}$.

\subsection{Proof}

Compare $\Sigma_{zW,yW}$ to $\Sigma^*_{z^*W,yW}$ in the case that $W$ are non-descendants of $Z$:

$$
\Sigma_{zW,yW} = \begin{bmatrix}
    \sigma_{zy} & \sigma_{zw_1} & \dots & \sigma_{zw_n}\\
    \sigma_{w_1y} & &  & \\
    \vdots &  & \Sigma_{W,W} & \\
    \sigma_{w_ny} &  &  & \\
\end{bmatrix}
$$

$$
\begin{aligned}
\Sigma^*_{z^*W,yW} &= \begin{bmatrix}
    \sigma^*_{z^*y} & \sigma^*_{z^*w_1} & \dots & \sigma^*_{z^*w_n}\\
    \sigma^*_{w_1y} & &  & \\
    \vdots &  & \Sigma^*_{W,W} & \\
    \sigma^*_{w_ny} &  &  & \\
\end{bmatrix}\\ &= \begin{bmatrix}
    \sigma^*_{zy} - \sum_i e_i\sigma^*_{p_iy} & \sigma^*_{zw_1} - \sum_i e_i\sigma^*_{p_iw_1} & \dots\\
    \sigma^*_{w_1y} & &   \\
    \vdots &  & \Sigma^*_{W,W}  \\
    \sigma^*_{w_ny} & & \\
\end{bmatrix}
\end{aligned}
$$

This was inserting the definition of auxiliary variable. We now use simple reasoning about paths
from wright's rules to conclude:

\begin{enumerate}
\item Since $W$ are non-descendants of $z$, $\Sigma^*_{W,W} = \Sigma_{W,W}$, since no paths
going through $z$ can go back to ancestors without crossing colliders.
\item Similarly, paths in $\sigma^*_{w_iy}$ do not cross any removed edges, and so is same for both graphs.
\item If $W$ are non-descendants of $z$, we have 
$$
\sigma^*_{zw_i} - \sum_i e_i\sigma^*_{p_iw_i} = \sigma_{zw_i}
$$
This can be seen by realizing that $\sigma_{zy}$ is the graph with the edges $e_i$ deleted. Since there are no paths passing through $e_i$ from the bottom (ie, $\sigma_{p_iw_i}$ does not have any paths through $e_i$), 
looking at wright's rules, we see that
we are simply removing all the paths through the deleted edges from $z$, meaning that no paths containing $e_i$ remain.
\item When y is a non-descendant of $z$, we have the same result. 
    $$
    \sigma^*_{zy} - \sum_i e_i\sigma^*_{p_iy} = \sigma_{zy}
    $$
\end{enumerate}

This is sufficient to prove the theorem as stated, since given the theorem's conditions, the two matrices are equal, meaning that $\sigma_{zy.W}^{G_{E-}}=\sigma_{z^*y.W}^{G^{E+}}$.

\subsubsection{What if $y$ is a descendant of $z$?}

The above theorem shows independencies behave as in the graph $G_{E-}$ when $y$ is not a descendant of $z$.
However, we use the AV $Z^*$ as an instrumental variable, which has $y$ as its descendant.
We therefore need to prove that the AV can be used as an instrumental variable even when $y$ is a descendant of $z$.

There are two differences from the above in this situation:

\begin{enumerate}
    \item $\sigma^*_{w_iy} = \sigma_{w_iy} + \gamma^{(e)}_{zw_i}\delta_{zy}$, since now the paths from $W$ to $y$ can cross removed edges,
    \item We also need to find the relationship between
    $$
\sigma^*_{zy} - \sum_i e_i\sigma^*_{p_iy}  \hspace{2em}\text{  and  }\hspace{2em} \sigma_{zy}
$$
\end{enumerate}

In the theorem's statement, the variance of $z$ in $G_{E-}$ was not specified. There are two possibilities. The first is having it be 1, and the other
is having $z$ in $G_{E-}$ retain the variance of $z^*$. Having the variance of $z$ be 1 causes some non-trivial changes in the graph, which require extra knowledge
of the values of directed paths to compute, so we specify that $z$ in $G_{E-}$ has the variance of $z^*$ exactly.

This means that to compare the two values in (2), we will need to expand both of them:

\vspace{1em}

\textbf{Expansion of $\sigma_{zy}$}

Whatever $\sigma_{z}^2$ is, we can decompose $\sigma_{zy}$ (ie: $\sigma_{zy}^{G_{E-}}$) using wright's rules for unnormalized models:

\begin{equation}
\label{basicdecomp}
\sigma_{zy} = \sigma_{z}^2 \delta_{zy} + \gamma_{zy}^{(-e)}
\end{equation}

First, notice that $\delta_{zy}$ is the same in $G^{E+}$ as in $G_{E-}$, since we did not remove any edges from descendants of $z$. The back-paths
are unaffected by the modified variance, but only the back-paths not taking edges $e_i$ are included (denoted as $-e$ in $\gamma_{zy}^{(-e)}$), 
since we are in the graph without these edges present.

Now, we want to compute $\sigma_z^2$. Since the graph was assumed to be normalized, we have $\E[x_ix_i] = 1 \forall x_i$ except for $z$ and its descendants. To get the variance of
$z$, we compute the variance of $z^*$ in $G^{E+}$ (denoted with $*$). This gives the variance of $z$ in the graph with the edges removed:
$$
\begin{aligned}
\sigma^{*2}_{z^*} &= \E^*[z^*z^*] = \E^*[(z - \sum_i e_i p_i)(z - \sum_i e_i p_i)]\\
&= \E^*[zz] - 2\sum_i e_i\E^*[zp_i] + \sum_i\sum_j e_ie_j \E^*[p_ip_j]\\
&= 1 - 2\sum_i e_i\sigma^*_{zp_i} + \sum_i\sum_j e_ie_j \sigma^*_{p_ip_j}\\
&= 1 - 2\sum_i e_i\left(\sigma_{zp_i} + \sum_j e_j \sigma_{p_ip_j}\right) + \sum_i\sum_j e_ie_j \sigma^*_{p_ip_j}\\
&= 1 - 2 \sum_i e_i \sigma_{zp_i} - \sum_i\sum_j e_ie_j \sigma_{p_ip_j}
\end{aligned}
$$

In the above, $\sigma^*_{zp_i}$ is just $\sigma_{zp_i}$ (from the graph with $e$ removed), plus all the paths that would have gone through $e$. 
The last step is because $\sigma^*_{p_ip_j} = \sigma_{p_ip_j}$, since paths between parents don't take any removed edges. Remember that $\sigma$ is in $G_{E-}$ and $\sigma^*$ is in $G^{E+}$

Substituting this as $\sigma^2_z$ in the decomposition of $\sigma_{zy}$ (eq \ref{basicdecomp}), we get:

\begin{equation}
    \label{AVexp1}
\sigma_{zy} = \delta_{zy}\left(1 - 2 \sum_i e_i \sigma_{zp_i} - \sum_i\sum_j e_ie_j \sigma_{p_ip_j}\right) + \gamma_{zy}^{(-e)}
\end{equation}

\vspace{1em}

\textbf{Expansion of $\sigma^*_{zy} - \sum_i e_i\sigma^*_{p_iy}$}

First look at  $\sigma^*_{zy}$ ($\sigma_{zy}$ in $G^{E+}$). Remember that $\gamma_{zy}^{(e)}$ is all back paths from $z$ to $y$ taking the AV edges $e$, 
and all back paths that do not take the AV edges as $\gamma^{-e}_{zy}$:

\begin{equation}
\sigma^*_{zy} = \delta_{zy} + \gamma_{zy}^{(e)} + \gamma_{zy}^{(-e)}
\end{equation}

Now, we take a closer look at the $\sum_i e_i\sigma_{p_iy}^*$ from the AV. Decompose $\sigma_{p_iy}^*$ as all the paths from $p_i$ through edge $e_j$, 
paths between $z$ and $p_i$ with edges $e$ removed, and finally
paths not going into $z$ at all:
$$
\sigma^*_{p_iy} = \delta_{zy} \sum_j e_j \sigma^*_{p_ip_j} + \delta_{zy}\sigma_{zp_i} + \frac{\gamma^{(e_i)}_{zy}}{e_i}
$$
...which makes:
$$
\sum_i e_i\sigma^*_{p_iy} = \delta_{zy}\sum_{i}\sum_{j} e_ie_j \sigma^*_{p_ip_j} + \delta_{zy}\sum_i e_i\sigma_{zp_i} + \gamma^{(e)}_{zy}
$$
giving us a result:

\begin{equation}
    \label{AVexp2}
\begin{aligned}
\sigma^*_{zy} - \sum_i e_i\sigma^*_{p_iy} &= \delta_{zy}\left(1 -\sum_i e_i\sigma_{zp_i} - \sum_{i}\sum_{j} e_ie_j \sigma_{p_ip_j} \right) + \gamma_{zy}^{(-e)}
\end{aligned}
\end{equation}

\vspace{1em}
\textbf{Final Result}

The two expansions in eq \ref{AVexp1} and \ref{AVexp2} can be put together, giving a single relation between them:
$$
\sigma^*_{zy} - \sum_i e_i\sigma^*_{p_iy} = \sigma_{zy} + \delta_{zy}\sum_i e_i \sigma_{zp_i}
$$

Plugging this into the two matrices:

$$
\det \Sigma_{zW,yW} = \det\begin{bmatrix}
    \sigma_{zy} & \sigma_{zw_1} & \dots & \sigma_{zw_n}\\
    \sigma_{w_1y} & &  & \\
    \vdots &  & \Sigma_{W,W} & \\
    \sigma_{w_ny} &  &  & \\
\end{bmatrix}
$$

$$
\begin{aligned}
\det \Sigma^*_{z^*W,yW} &= \det\begin{bmatrix}
    \sigma^*_{zy} - \sum_i e_i\sigma^*_{p_iy} & \sigma_{zw_1}  & \sigma_{zw_n} \\
    \sigma^*_{w_1y} &  & \\
    \vdots &  \Sigma_{W,W} & \\
    \sigma^*_{w_ny} &   & \\
\end{bmatrix} \\
&= \det\begin{bmatrix}
    \sigma_{zy} + \delta_{zy}\sum_i e_i \sigma_{zp_i} & \sigma_{zw_1}  & \sigma_{zw_n} \\
    \sigma_{w_1y} + \gamma^{(e)}_{zw_1}\delta_{zy} &  & \\
    \vdots &  & \Sigma_{W,W} \\
    \sigma_{w_ny} + \gamma^{(e)}_{zw_n}\delta_{zy} &  & \\
\end{bmatrix}
\end{aligned}
$$
\begin{equation}
= \det \Sigma_{zW,yW} + \delta_{zy}\det\begin{bmatrix}
    \sum_i e_i \sigma_{zp_i} & \sigma_{zw_1}  & \dots & \sigma_{zw_n} \\
    \gamma^{(e)}_{zw_1} & &  & \\
    \vdots &  & \Sigma_{W,W} & \\
    \gamma^{(e)}_{zw_n} &  &  & \\
\end{bmatrix}\label{AVfinal}\\ 
\end{equation}

If $\delta_{zy}=0$, we have the result for $y$ nondescendant of $z$.

\vspace{1em}
\textbf{AVs for IVs}

This last subsection is to ensure that IVs still work in this new situation. 

We have an IV, as defined by $z$ as the instrument, and $x\rightarrow y$ as the goal. By the requirements of IVs, we have:

$$
\begin{aligned}
\det \Sigma_{zW,yW} &= \det\begin{bmatrix}
        \sigma_{zy} & \sigma_{zw_1} & \dots & \sigma_{zw_n}\\
        \sigma_{w_1y} & &  & \\
        \vdots &  & \Sigma_{W,W} & \\
        \sigma_{w_ny} &  &  & \\
    \end{bmatrix} \\ &= \det\begin{bmatrix}
        \sigma^2_z\delta_{zx}\lambda + \gamma_{zx}\lambda + \gamma^{(-\lambda)}_{zy}  & \sigma_{zw_1} & \sigma_{zw_n}\\
        \sigma_{w_1x}\lambda + \sigma_{w_1y}^{(-\lambda)} &  & \\
        \vdots &  & \Sigma_{W,W} \\
        \sigma_{w_nx}\lambda + \sigma_{w_ny}^{(-\lambda)} &  & \\
    \end{bmatrix}\\ 
&= \lambda\det\Sigma_{zW,xW}  + \det\begin{bmatrix}
    \gamma^{(-\lambda)}_{zy}  & \sigma_{zw_1} &  \sigma_{zw_n}\\
    \sigma_{w_1y}^{(-\lambda)} & & \\
    \vdots &  & \Sigma_{W,W} \\
    \sigma_{w_ny}^{(-\lambda)} &  & \\
\end{bmatrix}\\ 
\end{aligned}
$$

Now, we will heavily exploit the fact that none of the relevant variables are descendants of $y$ to claim that the above determinant is 0
in the case of IVs (That is, we assume that $(z\indep y|W)$ in $G_{E-}$ with the ``goal'' edge $\lambda$ also removed. To do so, we look at the graph with $\lambda$ removed. In that case, we have a guarantee that $\delta^{(-\lambda)}_{zy}=0$,
so $\gamma^{(-\lambda)}_{zy} = \sigma^{(-\lambda)}_{zy}$. But this is the matrix for $\sigma_{zy.W}^{(-\lambda)}$, which we know is 0 by requirement, so:

\begin{equation}
    \label{matrixIV}
\det \Sigma_{zW,yW} = \lambda\det\Sigma_{zW,xW}
\end{equation}

since that is a requirement of the AV (when $\delta_{zy}=0$, $(z\indep y|W)$).
We therefore have, for our IV:

$$
\begin{aligned}
\frac{\sigma^*_{z^*y.W}}{\sigma^*_{z^*x.W}} &= \frac{\det \Sigma_{zW,yW} + \delta_{zy}\det\begin{bmatrix}
    \sum_i e_i \sigma_{zp_i} & \sigma_{zw_1}  & \dots & \sigma_{zw_n} \\
    \gamma^{(e)}_{zw_1} & &  & \\
    \vdots &  & \Sigma_{W,W} & \\
    \gamma^{(e)}_{zw_n} &  &  & \\
\end{bmatrix}}{\det\Sigma_{zW,xW} + \delta_{zx}\det\begin{bmatrix}
    \sum_i e_i \sigma_{zp_i} & \sigma_{zw_1}  & \dots & \sigma_{zw_n} \\
    \gamma^{(e)}_{zw_1} & &  & \\
    \vdots &  & \Sigma_{W,W} & \\
    \gamma^{(e)}_{zw_n} &  &  & \\
\end{bmatrix}} \\
&= \frac{\lambda\left(\det\Sigma_{zW,xW} + \delta_{zx}\det[\cdots] \right)}{\det\Sigma_{zW,xW} + \delta_{zx}\det[\cdots] } = \lambda
\end{aligned}
$$

As a side note, the above proof is also valid if $x=z$, so we can use $x$ as the AV itself.

\section{Additional Proofs}
\subsection{Conditional Edge Lemmas}
First we give 3 lemmas which are used extensively in the coming proofs. They are referred to as the \textit{Conditional Edge Lemmas}, or $CEL$. 

For convenience, we will use a shorthand notation of $\sigma_{xy.W}^G = \sigma(x,y|W)$ in the graph $G$.\footnote{Different graphs can have different covariances of the same variables. Since each graph is defined by SEMs, the effect of adding or removing variables to equations (edges) is well-defined in terms of the covariances.} 

\begin{elemma}
\label{thm:conditionalGraphParents}
Given variables $x,y$, a conditioning set $W$, and defining $p_i = Pa(x)_i$, then $\sigma_{xy.W}=\sum_i \alpha_i \sigma_{p_iy.W} + \sigma_{u_xy.W}$, where $\alpha_i$ as the structural parameter for the edge
between $p_i$ and $x$,
and $u_x$ is the error term of $x$.
\end{elemma}

\begin{proof}
Let $\{w_1,...,w_n\} = W$. By definition of conditional covariance,

$$
\sigma_{xy\cdot W} = \E[\eta_{x\cdot W}\eta_{y\cdot W}]
$$

where $\eta_{x\cdot W}$ is the residual:

$$
\eta_{x\cdot W} = x- \sum_i\beta_i w_i
$$

with the $\beta_i$ as regression coefficients. Note that by definition of the residual $\E[w_i\eta_{y\cdot W}] = 0$, i.e. the covariance of a residual with any of its subtracted variables is 0.

$$
\begin{aligned}
\sigma_{xy\cdot W} = \E[\eta_{x\cdot W}\eta_{y\cdot W}] &= \E\left[(x- \sum_i\beta_i w_i)\eta_{y\cdot W}\right]\\ &= \E\left[x\eta_{y\cdot W}\right]
\end{aligned}
$$

Expanding the definition of $x$:

$$
\begin{aligned}
\E\left[x\eta_{y\cdot W}\right] &= \E\left[\left(\sum_i \alpha_i p_i + u_x\right)\eta_{y\cdot W}\right] \\&= \sum_i \alpha_i \E\left[ p_i\eta_{y\cdot W}\right] + \E\left[ u_x \eta_{y\cdot W}\right] 
\end{aligned}
$$

We now subtract the regression coefficients for each variable, since we are subtracting 0 in the expectation (covariance of a residual with its subtracted variables is 0), turning the $p_i$ back into residuals.

$$
\begin{aligned}
&\sum_i \alpha_i \E\left[ p_i\eta_{y\cdot W}\right] + \E\left[ u_x \eta_{y\cdot W}\right]\\ &= \sum_i \alpha_i \E\left[ \eta_{p_i\cdot W}\eta_{y\cdot W}\right] + \E\left[ \eta_{u_x\cdot W} \eta_{y\cdot W}\right]\\ & = \sum_i \alpha_i \sigma_{p_iy.W} + \sigma_{u_xy.W}
\end{aligned}
$$
\end{proof}

\begin{elemma}
\label{thm:conditionalEdgeRemove}
Given a conditional covariance $\sigma_{xy.W}$ in graph $G$, labeled as $\sigma_{xy.W}^G$, and a set of directed edges $E$, where $G_{E-}$ is the graph $G$ with edges $E$ removed, 
if $(W\cup\{x,y\})\cap Desc(Head(E)) = \emptyset$, then $\sigma_{xy.W}^G = \sigma_{xy.W}^{G_{E-}}$.
\end{elemma}

\begin{proof}
As done in CEL \ref{thm:conditionalGraphParents}, we directly use the definition of conditional covariance in terms of regression:
$$
\sigma_{xy\cdot W} = \E[\eta_{x\cdot W}\eta_{y\cdot W}] \ \ \ \text{where}\ \ \  \eta_{x\cdot W} = x- \sum_i\beta_i w_i
$$

$\beta$ is computed by minimizing the squared residual:

$$
\begin{aligned}
\E[\eta_{x\cdot W}&\eta_{x\cdot W}] = \E\left[(x- \sum_i\beta_i w_i)^2\right] \\
&= \E[x^2] - \sum_i \beta_i \left(2\E[xw_i] - \sum_j \beta_j \E[w_iw_j] \right)
\end{aligned}
$$

This equation holds in all graphs. We will show that the expectation terms of the equation, and hence the resulting values of $\beta$ after performing regression
are the same in $G$ as they are in $G_{E-}$. 

Since $(W\cup\{x,y\})\cap Desc(Head(E)) = \emptyset$, we know that $x$ and $w_i$ are both non-descendants of the removed edges in $G_{E-}$, so
the $\E[xx]$, and all $\E[xw_i]$ and $\E[w_iw_j]$
terms can be directly expanded in terms of their ancestors, which are the same for both $G$ and $G_{E-}$, and have the same underlying error distribution and covariances\footnote{we are working in DAGs - non-recurrent models}. This means that
these terms must be equal in $G$ and $G_{E-}$.

Another way to reason about this is to use Wright's rules of path analysis. The terms $\E[xw_i]$ can be written in terms of paths between $x$ and $w_i$. For a path
to cross a removed edge, it would need to cross a collider in order to leave the descendants of the edge, and get to the goal node. This means that the valid paths are the same for both graphs, giving the equations
used to solve for $\beta$ identical expectation coefficients.

We can now expand out the value of $\sigma_{xy\cdot W}$ the same way in both graphs:

$$
\begin{aligned}
\sigma_{xy\cdot W} &= \E[\eta_{x\cdot W}\eta_{y\cdot W}] = \E[\eta_{x\cdot W}y] \\&= \E\left[(x- \sum_i\beta_i w_i)y\right] = \E[xy] - \sum_i \beta_i \E[w_iy]
\end{aligned}
$$

We have showed that $\beta$ are the same in both graphs, and we use the same reasoning to conclude that $\E[w_iy]$ and $\E[xy]$ must be equal in $G$ and $G_{E-}$. 
Therefore, since all terms in the equation are the same in both graphs, $\sigma_{xy\cdot W}^G = \sigma_{xy\cdot W}^{G_{E-}}$.

\end{proof}

\begin{elemma}
\label{thm:conditionalErrorRemove}
Given a conditional error covariance $\sigma_{u_xy.W}^G$, and a set of directed edges $E$,
if $(W\cup\{y\})\cap Desc(Head(E)) = \emptyset$, then $\sigma_{u_xy.W}^G = \sigma_{u_xy.W}^{G_{E-}}$.
\end{elemma}

The main difference between this and CEL \ref{thm:conditionalEdgeRemove}, is that we operate on $u_x$ (the error term of $x$), which allows $x$ to be a descendant of $Head(E)$.

\begin{proof}
We proceed in the same fashion as in CEL \ref{thm:conditionalEdgeRemove}. By the definition of conditional error covariance:

$$
\begin{aligned}
\sigma&_{u_xy.W} = \E\left[ \eta_{u_x\cdot W} \eta_{y\cdot W}\right] = \E\left[ u_x \eta_{y\cdot W}\right] \\&=  \E\left[ u_x \left(y - \sum_i \beta_i w_i\right)\right] 
= \E[u_xy] - \sum_i \beta_i \E[u_xw_i]
\end{aligned}
$$

Using the reasoning from CEL \ref{thm:conditionalEdgeRemove}, we know that $\beta_i$ are the same for $G$ and $G_{E-}$. Once again, expanding $y$ and $w_i$ to their ancestors, which have no edges removed,
we get the same distributions for both graphs, meaning that the expectations are also equal.

This can also be seen intuitively in terms of Wright's rules when $x$ is not an ancestor of $y$. In that case, $\E[u_xy]$ represents all paths from x to y starting with a bidirected edge (half-treks).
If such a path were to be different in the two graphs, it would need to cross a deleted edge. But to do that, it would have to cross a collider. If $x$ is an ancestor of $y$, then we will additionally have an $\E[u_xu_x]$ term in our expansion,
which is the same for both graphs.
\end{proof}

\subsection{Auxiliary and Quasi-Instrumental Sets}


\begin{restatable}{suppdefinition}{aiv}
\label{def:aiv}
Given a linear SEM with graph G, a set $E_Z$ of known coefficients, 
and a set of structural coefficients $\alpha = \{\alpha_1,\alpha_2,...,\alpha_k\}$, 
the set $Z=\{z_1,...,z_k\}$ generates an auxiliary instrumental set if there exist triples 
$(z_1,W_1,p_1),...,(z_j,W_j,p_k)$ such that:

\begin{enumerate}
\item For $i=1,...,k$, either:
\begin{enumerate}
\item the elements of $W_i$ are non-descendants of $y$, and $(z_i\indep y|W_i)_{G_{E}}$ where $G_{E}$ is the graph obtained by deleting the edges $E$ from $G$.
\item the elements of $W_i$ are non-descendants of $z_i$ and $y$, and $(z_i\indep y|W_i)_{G_{E\cup E_{z_i}}}$ where $G_{E\cup E_{z_i}}$ is the graph obtained by deleting the edges $E$, $E_Z\cap (Inc(z_i))$ from $G$.
\end{enumerate}
\item for $i=1,...,k$, $p_i$ is an unblocked path between $z_i$ and $x_i$, not blocked by $W_i$, where $x_i=He(\alpha_i)$, 
\item the set of paths $\{p_1,...,p_k\}$ has no sided intersection
\end{enumerate}
\end{restatable}

\begin{supptheorem}
\label{thm:auxid}
If there exists an auxiliary instrumental set for structural coefficients $\{\alpha_1,\alpha_2,...,\alpha_k\}$, then the coefficients are identifiable.
\end{supptheorem}

\begin{proof}
Here, we will do the same exact proof as for standard AV, using $\tau_z$ to represent the extra determinant with respect to $z$.

$$
\begin{aligned}
\sigma^*_{z^*y.W} &= \sigma_{zy.W} + \delta_{zy}\tau_z\\& = \sum_i \lambda_i \sigma_{zx_i.W} + \tau_z \sum_i \lambda_i \delta_{zx_i}\\ & = \sum_i \lambda_i \sigma^*_{z^*x_i.W}
\end{aligned}
$$

The above equation shows that the system of linear equations used for instrumental sets is also valid for AVs.
To show that this system can be solved, we modify \citet{brito:pea02a}'s proof of instrumental sets. The modifications span multiple lemmas, therefore the full proof is given as appendix \ref{sec:supptheorem} of this document (below).
\end{proof}

\stepcounter{theorem}

\begin{restatable}{theorem}{aqivid}
\label{thm:aqivid}
If $Z^*$ is a quasi-instrumental set for $E$, then the coefficients $E$ are identifiable.
\end{restatable}

\begin{proof}
Suppose we have a quasi-instrumental set for $E=\{e_1,...,e_k\}$ with $Z^* = \{z_1,...,z_k\}$ ($z_i$ is referring to the auxiliary variable itself rather than its generator). We know that this set is solvable in the graph $G_{E_y}$,
where the graph is obtained by deleting the edges $T = E_{Z}\cap Inc(y)$ from $G$, since it is an auxiliary instrumental set for the graph. 

Let the parameters connecting $t\in T$ to $y$ be $\gamma$. 
 Let $T'$ be all incident edges to $y$ that are not in $T$ or $E$.
That is, $T' = Inc(y)\setminus (E\cup T)$ (and let the associated structural parameters be $\gamma'$). Finally, let $X$ be $Tail(E)$.

We will show that there exists a solution by explicitly constructing the linear equations to be solved for the parameters. For each $z_i$, we generate an equation:

$$
\begin{aligned}
\sigma_{z_iy^*.W_i} &= \sigma_{z_iy.W_i} - \sum_j \gamma_j \sigma_{z_it_j.W_i}\\
&= \sum_j e_j\sigma_{z_ix_j.W_i} + \sum_j \gamma'_j \sigma_{z_it'_j.W_i} + \sigma_{z_iu_y.W_i}\\
\end{aligned}
$$
We will use the Conditional Edge Lemmas to move the last two terms into the graph $G_{E-\cup E_{y}}$, where these terms are equal to $\sigma_{yz_i.W_i}^{G_{E-\cup E_{y}}}$.
We notice that the second term in the resulting equation must be 0, since by definition of quasi-IV $(z\indep y|W)_{G_{E-\cup E_{y}}}$
$$
\begin{aligned}
\sigma_{z_iy^*.W_i} &= \sum_j e_j\sigma_{z_ix_j.W_i} + \sum_j \gamma'_j \sigma_{z_it'_j.W_i} + \sigma_{z_iu_y.W_i}\\
&= \sum_j e_j\sigma_{z_ix_j.W_i} + \sigma_{yz_i.W_i}^{G_{E-\cup E_{y}}}\\
&= \sum_j e_j\sigma_{z_ix_j.W_i}
\end{aligned}
$$

We now have a system of linear equations, one for each $z_i$, in terms of the $e_i$. The system is in the form $Ae=b$. The $A$ matrix is full rank, because by the Conditional Edge Lemmas, all terms in the matrix
are the same as their counterparts in $G_{E-\cup E_{y}}$. We know that if we find a quasi-instrumental set, then there exists at least one quasi-instrumental set $Z^*$ which makes this matrix full rank. We proved the existence of such a set in supplementary theorem \ref{thm:auxid}. That is, we showed that if one auxiliary set exists, we can always construct another for $E$, for which the above matrix is full rank, and thus invertible. For details, see proof of Supplemental Theorem 1.

\end{proof}

\setcounter{corollary}{0}

\begin{corollary}
\label{cor:aux}
Given a linear SEM with graph $G$, $z^*$ is a quasi-IV for $\alpha$ given $W$ if $W$ does not contain any descendants of $z$, and $z$ is an IV for $\alpha$ given $W$ in $G_{E_z \cup E_y-}$, where $E_z\subseteq Inc(z)$ and $E_y\subseteq Inc(y)$ are sets of edges whose coefficient values are known.

\end{corollary}

\begin{proof}
Let IV-(i), IV-(ii), and IV-(iii) denote conditions (i)-(iii) of Lemma 1 in \citet{pearl2011parameter} and let $\alpha$ be the coefficient of edge $(x,y)$.  We need to show that IV-(i), IV-(ii), and IV-(iii) hold in $G^{E+}$.  Since $z$ is an IV for $\alpha$ given $W$ in $G_E$, it must be the case that $z^*$ satisfies IV-(i) and IV-(iii) in $G^{E+}$.  Now, it remains to be shown that $(z^* \indep y|W)_{G^{E+}_\alpha}$.  Theorem \ref{thm:theoremOne} tells us that if $(z\indep y|W)_{G_{E\cup \{\alpha\}}}$ and $W\cup \{y\}$ does not contain descendants of $z$ in $G_{E\cup \{\alpha\}}$, then $(z^* \indep y|W)_{G^{E+}_\alpha}$.  By assumption, $W$ does not contain any descendants of $z$.  $y$ also cannot be a  descendant of $z$ in $G_{E\cup \{\alpha\}}$.  If $y$ were a descendant of $z$, then it would not be possible to block the path from $z$ to $y$ using $W$, which does not contain any descendants of $z$.  
\end{proof}

\begin{theorem}
Given an arbitrary linear causal model, if a set of coefficients is identifiable using the g-HT algorithm, then it is identifiable using $qID$.  Additionally, there are models that are not identified using the g-HT algorithm, but identified using $qID$.
\end{theorem}

\begin{proof}
Proved in the paragraph preceding theorem statement in paper
\end{proof}

\begin{restatable}{theorem}{overid}
\label{thm:overid}
Let $Z$ be a quasi-IV set for structural coefficients $\alpha=\{\alpha_1,...,\alpha_k\}$ and $E$ be a set of known edges. If there exists a node $s$ satisfying the conditions listed below, then $\alpha$ is overidentified.
\begin{enumerate}
\item $s\notin Z$
\item There exists an unblocked path between $s$ and $y$ including an edge in $\alpha$
\item There exists a conditioning set $W$ that does not block the path $p$, such that either:
  \begin{enumerate}
  \item the elements of $W$ are non-descendants of $y$, and $(s\indep y|W)_{G_{\alpha\cup E_y}-}$, where $E_y = E \cap Inc(y))$
  \item the elements of $W$ are non-descendants of $s$ and $y$, and $(s\indep y|W)_{G_{\alpha \cup E_{s}\cup E_y-}}$ where $E_s = E \cap Inc(s)$.
  \end{enumerate}
\end{enumerate}
\end{restatable}

\begin{proof}
In the proof of theorem \ref{thm:aqivid}, we generated a full-rank set of linear equations, where each equation had the form:
$$
\sigma_{z_iy^*.W_i} = \sum_j e_j\sigma_{z_ix_j.W_i}
$$
We can generate a set of linear equations of the form $Ae=b$, using the above.

Similarly, we can use the parameter $s$ to generate another single equation in the given form: $a_s e = b_s$. 
Now, if $Z_E$ is a full auxiliary set, then $A$ is invertible, so we get $e=A^{-1}b$, giving us the overidentifying constraint $a_sA^{-1}b = b_s$.
\end{proof}

\begin{theorem}
\label{thm:aCI-subsumed}
Let $z^*=z-e_1 t_1 - ... - e_k t_k$ and suppose there does not exist $W$ such that $(z\indep y|W)_G$. There exists $W$ such that $W\cap De(z) = \emptyset$ and $(z^*\indep y|W)$ is non-vacuous if and only if $y$ satisfies the conditions of Theorem \ref{thm:overid} for $E=\{e_1 ,... , e_k\}$.
\end{theorem}
 \begin{proof}
	($\implies$)
	First, we show that $y$ satisfies $(ii)$ and $(iii)$ of Theorem \ref{thm:overid}. Since $z\not\indep y|W$ but $z^* \indep y |W$ there must exist a path from $y$ to $z$ that goes through $E$ and $(ii)$ is satisfied. Next, $(z^* \indep y|W)$ implies that $(z\indep y|W)_{G_E-}$ so $(iii)$ is satisfied.
	
	Next, we show that there exists $T=\{t_1 , ..., t_k\}$, $y\notin T$, such that $T$ is an quasi-IV set for $E$ so that (i) is satisfied. Since $(z^*\indep y|W)$ is not vacuous, $E$ is identified in $G^{'}$, the graph where a directed edge from $y$ to $z$, called $e_{yz}$, is added. As a result, there exists $T$ such that $y\notin T$ and $T \cup \{y\}$ is a quasi-IV set for $E\cup \{e_{yz}\}$. It follows that $T$ is a quasi-IV set for $E$. 
	
	($\impliedby$)
	Let $T$ be the quasi-IV set for $E$ that does not include $y$. (iii) implies that there exists $W$ such that $(y\indep z|W)_{G_E-}$, and, since $E$ is identifiable using $T$, $(z^* \indep y|W)$. Finally, this independence cannot be vacuous since $T\cup \{y\}$ is a quasi-IV set for $E\cup \{e_{yz}\}$ in $G^{'}$. 
\end{proof}
    
\subsection{Identification and z-Identification Algorithm}

Two algorithms are given for finding Quasi-Instrumental Sets. The first version does not consider IVs that are conditioned
on descendants of $z$, whereas the second version is more computationally expensive (still polynomial if $k$ is bounded), but is able to find any quasi-instrumental set
if such exists.

In $FindQIS$, we make extensive use of $TestQIS$, which is a modification of $TestGeneralIVs(G,X,Y,Z)$ from \citet{zander:etal16}. Our version has 2 extra arguments,
and replaces the first 4 lines of $TestGeneralIVs$ such that we can search for both auxiliary instruments ($Aux=1$) and standard instrumental variables ($Aux=0$).

\begin{algorithm}[H]
\caption{Modified version of $TestGeneralIVs$ from \citet{zander:etal16} for use with $findAuxIS$ }
\begin{algorithmic}
\Function{TestQIS}{G,X,Y,Z,IDEdges,Aux}
  \For{$i$ in $1,...,|Z|$}
    \If{$Aux_i == 1$}
      \State $W_i \gets $ a nearest separator for $(Y,Z_i)$ 
      \State \ \ \ \ \ in $G_{E\cup E_{z_i}\cup E_y}$, where $E_{z_i}$ 
      \State \ \ \ \ \ is $IDEdges \cap Inc(z_i)$
      \If{$W_i = \bot \vee (W_i\cap De(Y)) \ne \emptyset 
      \vee (W_i\cap De(z_i)) \ne \emptyset$}
        \State \Return $\bot$
      \EndIf
    \Else
      \State $W_i \gets $ a nearest separator for $(Y,Z_i)$ in $G_{E\cup E_y}$
      \If{$W_i = \bot \vee (W_i\cap De(Y)) \ne \emptyset$}
        \State \Return $\bot$
      \EndIf
    \EndIf
  \EndFor
  \State continue algorithm $TestGeneralIVs$ starting 
  \State from second for loop. 
  \State Instead of returning $False$, return $\bot$, 
  \State and instead of returning $True$, return $W$.
\EndFunction
\end{algorithmic}
\end{algorithm}

\begin{algorithm}[H]
\caption{Finds a quasi-instrumental set (without conditioning on descendants in IVs)}
\begin{algorithmic}
\Function{FindQISBasic}{E,G,IDEdges}
  \ForAll{$Z\subset V\setminus \{y\}$ of size $|E|$}
    \State $W \gets $\Call{TestQIS}{G,$Ta(E)$,$Head(E)$,$Z$,1}
    \If{$W\ne \bot$}
    \State \Return $(Z,W)$
    \EndIf
  \EndFor
  \State \Return $\bot$
\EndFunction
\end{algorithmic}
\end{algorithm}

\begin{algorithm}[H]
\caption{Finds a quasi-instrumental set for $E$ in $G$, given a set IDEdges of identified edges.}
\label{alg:findqis}
\begin{algorithmic}
\Function{FindQIS}{E,G,IDEdges}
\ForAll{$Z\subset V\setminus \{y\}$ of size $|E|$}
    \ForAll{Aux $\in \{0,1\}^{|E|}$}
      \State $W \gets $\Call{TestQIS}{$G$,$Ta(E)$,$Head(E)$,$Z$,Aux}
      \If{$W\ne \bot$}
      \State \Return $(Z,W)$
      \EndIf
    \EndFor
  \EndFor
  \State \Return $\bot$
\EndFunction
\end{algorithmic}
\end{algorithm}

\begin{algorithm}[H]
\label{alg:eiv}
\caption{Tests whether $w$ fits the conditions of theorem \ref{thm:aCI-subsumed}}
\begin{algorithmic}
\Function{IsEIV}{w,E,G,IDEdges}
	\State Let $G'$ be the graph $G$ modified such that $E$ are 
  	\State removed, and each node in $Tail(E)$ has an edge 
    \State added to a newly created node $n$, which has an 
    \State edge to $Head(E)$
	\ForAll{Aux $\in \{0,1\}$}
      \State $W \gets $\Call{TestQIS}{$G'$,$\{n\}$,$Head(E)$,$\{w\}$,Aux}
      \If{$W\ne \bot$}
      \State \Return $(Z,W)$
      \EndIf
    \EndFor
    \State \Return $\bot$
\EndFunction
\end{algorithmic}
\end{algorithm}

$$
$$

The function $IsEIV$, is a slight modification of $FindQIS$ that makes the subset a full auxiliary set in a graph modified so that the full set of $E$ has directed edges to a single node, instead of $y$, so that this node can be a new set $E'$ of size 1.

\label{sec:ivav}

\section{Proof of Supplemental Theorem \ref{thm:auxid}}
\label{sec:supptheorem}

We build upon the proof given in \citet{brito:pea02a} to show that auxiliary instrumental sets are identifiable.

\subsection{Generalized Instrumental Sets}

We will use the definition of generalized instrumental set directly from \citet{brito:pea02a}'s paper.

\begin{mydef}
\label{def:iv}
The set $Z$ is said to be an instrumental set relative to $X$ and $Y$ if we can find triples $(Z_1,W_1,p_1),...,(Z_n,W_n,p_n)$ such that for $i=1,...,n$
\begin{enumerate}
\item $Z_i$ and the elements of $W_i$ are non-descendants of $Y$; and $p_i$ is an unblocked path between $Z_i$ and $Y$ including edge $X_i\rightarrow Y$
\item Let $\bar G$ be the causal graph obtained from $G$ be deleting edges $X_1\rightarrow Y$, $X_n\rightarrow Y$. Then $W_i$ d-separates $Z_i$ from $Y$ in $\bar G$, but $W_i$ does
not block path $p_i$
\item For $1 \le i \le j \le n$, $Z_j$ does not appear in path $p_i$, and, if paths $p_i$ and $p_j$ have a common variable $V$, then both $p_i[V\sim Y]$ and $p_j[Z_j\sim V]$ point to $V$.
\end{enumerate}
\end{mydef}

The third property is written here in the same way it is written in \citet{brito:pea02a}. We used 
\textit{$p_i$ and $p_j$ do not have any sided intersection} instead. 
The two methods for writing the property are equivalent, 
meaning that there exists a set satisfying the \citet{brito:pea02a} definition iff there exists a set satisfying our definition 
(note that the two sets might be different). This is proved in Appendix \ref{sec:equivIV} of this document.

\subsection{Auxiliary Instrumental Sets}

We perform an equivalent translation to the definition of Auxiliary Instrumental Set:

\begin{mydef}
\label{def:aiv}
Given a linear SEM with graph G, a set $E_Z$ of known coefficients, 
and a set of structural coefficients $\alpha = \{\alpha_1,\alpha_2,...,\alpha_k\}$, 
the set $Z=\{z_1,...,z_k\}$ generates an auxiliary instrumental set if there exist triples 
$(z_1,W_1,p_1),...,(z_j,W_j,p_k)$ such that:

\begin{enumerate}
\item For $i=1,...,k$, either:
\begin{enumerate}
\item the elements of $W_i$ are non-descendants of $y$, and $(z_i\indep y|W_i)_{G_{E}}$ where $G_{E}$ is the graph obtained by deleting the edges $E$ from $G$.
\item the elements of $W_i$ are non-descendants of $z_i$ and $y$, and $(z_i\indep y|W_i)_{G_{E\cup E_{z_i}}}$ where $G_{E\cup E_{z_i}}$ is the graph obtained by deleting the edges $E$, $E_Z\cap (Inc(z_i))$ from $G$.
\end{enumerate}
\item for $i=1,...,k$, $p_i$ is an unblocked path between $z_i$ and $y$, not blocked by $W_i$, including the edge $(x_i,y)$
\item For $1 \le i \le j \le n$, $Z_j,Z'_j$ does not appear in path $p_i$, and, if paths $p_i$ and $p_j$ have a common variable $V$, then both $p_i[V\sim Y]$ and $p_j[Z_j\sim V]$ point to $V$.
\end{enumerate}
\end{mydef}

\subsection{Auxiliary Sets generate Generalized Instrumental Sets}

\begin{lemma}
\label{thm:auxis}
If there exists an auxiliary instrumental set for structural coefficients $\{\alpha_1,\alpha_2,...,\alpha_k\}$, then there exists a generalized instrumental set for the coefficients
in $G^{E+}$.
\end{lemma}

\begin{proof}
We will denote conditions 1 through 3  of Supplemental Definition \ref{def:aiv}
as AIV 1-3, respectively. We will denote the conditions of Definition \ref{def:iv} as GIV 1-3. This proof will proceed by showing that we can generate a generalized instrumental set in $G^{E+}$ using the auxiliary set.

We have defined $G^{E+}$ as the graph where all possible auxiliary variables have been added.  For each $z_i$ in $Z$:
\begin{enumerate}
\item if $z_i$ satisfies AIV 1a, then $(z_i \indep y|W)_{G^{E+}}$, because the added node $z_i^*$ is a collider for any possible paths going through it. If $z_i$ satisfies AIV 1b, then $(z^*_i\indep y |W)_{G^{E+}}$ using Theorem \ref{thm:theoremOne}. Therefore, GIV 1 is satisfied.
\item If AIV 2 is satisfied, then GIV 2 follows directly if AIV 1a is satisfied. If AIV 1b is satisfied, we can extend the path from AIV 2 with the edge $z_i^* \leftarrow z_i$. Since $z^*_i$ is unblocked, this new path will satisfy GIV 2.
\item If AIV 3 is satisfied, then the paths ($p_i$) constructed in the previous part will not have sided intersection We might have added the edge $z_i^* \leftarrow z_i$ which makes $z_i$ in $Left(p_i)$, but the original $z_i$ was in $Left(p_i)$ already by the definition of $Left$. Furthermore, $z^*_i$ is a collider, so it could not be part of any other variable's path. This means GIV 3 is satisfied.
\end{enumerate}

Since all of the conditions necessary for definition \ref{def:iv} are satisfied,
we have constructed a generalized instrumental set for $G^{E+}$.
\end{proof}

\subsection{Identifiability of Generalized IVs does NOT imply ID of Aux IVs}

In generalized IVs, it is assumed that all edges in the graph have independent structural parameters. 
When using auxiliary variables, the edges incoming to the auxiliary variable are repeating the structural
parameters found elsewhere in the graph. This invalidates the assumption of independence implicit in Definition \ref{def:iv}.

Furthermore, it turns out that in proving the identifiability of coefficients
from a generalized instrumental set, \citet{brito:pea02a} generated another instrumental set, with a special property. 
They argued that this new set still satisfied the conditions of Definition \ref{def:iv}. With auxiliary variables, it is not clear that it is possible to modify the auxiliary set, since the independence
properties of the variables are different, since $Z^*$ has coefficients cancel only after subtracting the auxiliary paths.

We will show that \citet{brito:pea02a}'s proof can be modified to show identifiability in auxiliary instrumental sets.

\subsubsection{Preliminaries}

First, we will quickly review the relevant portions of the proof of generalized IVs.

\begin{lemma} (Partial Correlation Lemma, \citet{brito:pea02a})
The partial correlation $\rho_{12.3...n}$ can be expressed as the ratio:

$$
\rho_{12.3...n} = \frac{\phi(1,2,...,n)}{\psi(1,3,...,n)\psi(2,3,...,n)}
$$

where $\phi$ and $\psi$ are functions satisfying the following conditions:
\begin{enumerate}
\item $\phi(1,2,...,n) = \phi(2,1,...,n)$
\item $\phi(1,2,...,n)$ is linear on correlations $\rho_{12},\rho_{32},...\rho_{n2}$ with no constant term
\item The coefficients of $\rho_{12},\rho_{32},...\rho_{n2}$ in $\phi(1,2,...,n)$ are polynomials on the correlations among $Z$, $W_i,...$. Furthermore, the coefficient of $\rho_{12}$
has its constant term $=1$, and the coefficients of $\rho_{32},...,\rho_{n2}$ are linear on the correlations $\rho_{13},\rho_{14},...,\rho_{1n}$ with no constant term
\item $(\psi(i_1,...,i_{n-1}))^2$ is a polynomial on the correlations among variables $Y_{i_1},...,Y_{i_{n-1}}$ with constant term $=1$.
\end{enumerate}

\end{lemma}

With this lemma in hand, we will outline how \citet{brito:pea02a} showed that IVs are identifiable by restating the lemmas, and giving 2 sentence descriptions of how they were proved.

\begin{lemma}
\label{lemma:nosharing}
(Lemma 2, \citet{brito:pea02a}) WLOG, we may assume that for $1\le i < j \le n$, paths $p_i$ and $p_j$ do not have any common variable other than (possibly) $Z_i$.
\end{lemma}

\begin{proof} (Outline)
Suppose not. That is, suppose that paths $p_i$ and $p_j$ have a variable in common other than $Z_i$. Call this variable $V$. We can now generate a new instrumental set using $V$ instead
of $Z_i$. That is, if there exists a common variable, we can generate a new instrumental set, where this variable is $Z_i$. This new instrumental set conforms to the definition \ref{def:iv}.
This is proved by showing that since $Z_i$ is independent of $Y$ given $W_i$, $V$ must also be independent, since there is a directed, unblocked, path from $V$ to $Z_i$.
\end{proof}

\begin{lemma}
\label{lemma:nounblocked}
For all $1\le i \le n$, there exists no unblocked path between $Z_i$ and $Y$, different from $p_i$, which includes edge $X_i\rightarrow Y$, and is composed only of edges from $p_1,...,p_i$.
\end{lemma}

\begin{proof} (Outline)
By contradiction - suppose such a path exists, then since it is different from $p_i$, it must contain edges from $p_1,...,p_{i-1}$. But all such paths that intersect with $p_1$ will do so at a collider.
\end{proof}

\begin{lemma}
For all $1\le i \le n$, there exists no unblocked path between $Z_i$ and some $W_i$, composed only of edges from $p_1,...,p_i$.
\end{lemma}

\begin{lemma}
For all $1\le i \le n$, there exists no unblocked path between $Z_i$ and $Y$, including edge $X_j \rightarrow Y$, with $j< i$, composed only of edges from $p_1,...,p_i$.
\end{lemma}

These two lemmas use the same proof method as lemma \ref{lemma:nounblocked}, and the proofs are omitted. Using these 3 lemmas, \citet{brito:pea02a} proved that the determinant of the linear system
is a non-trivial polynomial, whose zeros have lebesgue measure zero.

\subsubsection{Proof Modification for Auxiliary Variables}

The above lemmas are the only thing which needs to be modified to work with Auxiliary Variables. Lemma \ref{lemma:nosharing} needs to be modified to take into account the fact that Auxiliary Variables
have different independence properties, whereas lemma \ref{lemma:nounblocked} and its siblings need to take into account that edges are repeated in our graph.

\begin{lemma}
\label{lemma:auxnosharing}
WLOG, we may assume that for $1\le i < j \le n$, paths $p_i$ and $p_j$ do not have any common variable other than (possibly) $Z_i$ or $Z'_i$ (parent of $Z_i$ if it is an auxiliary variable).
\end{lemma}

\begin{proof}
Assume that paths $p_i$ and $p_j$ have some variables in common, different from $Z_i$ (which might be an auxiliary variable). 
Let $V$ be the closest variable to $X_i$ in path $p_i$ which also belongs to path $p_j$. We show that after replacing $(Z_i,W_i,p_i)$ with $(V,W_i,p_i[V\sim Y])$, definition \ref{def:aiv} still holds.

From (3), changed to be in the format of GIVs, the subpath $p_i[V \sim Y]$ must point to $V$. Since $p_i$ is unblocked, subpath $p_i[Z_i\sim V]$ must be a directed path from $V$ to $Z_i$.
Furthermore, if $Z_i$ is an auxiliary variable, $p_i$ did not cross any of the subtracted edges, since the path was found in a graph with these edges removed.

At this point, if the variable $Z_i$ is not an auxiliary variable, the 3 conditions hold:
\begin{enumerate}
\item (a) is satisfied, since $p_i[Z_i\sim V]$ is a directed path from $V$ to $Z_i$, so if $V$ is descendant of $y$ then $Z_i$ is a descendant of $y$. Similarly, if $(v\not \indep y|W_i)_{G_{E}}$, 
then $(z_i\not \indep y|W_i)_{G_{E}}$, because if $W_i$ does not d-separate $V$ from $y$, then since $W_i$ are not blocking $p_i$, we can generate a path from $Z_i$ to $y$ through $V$.
\item Since the path from $V$ to $Y$ is a subpath of the path $Z_i \sim Y$, the path is unblocked.
\item The path from $Z_i$ to $y$ must have $V\in Left$, since $p_i[Z_i\sim V]$ is a directed path. Therefore, the new path has no sided intersection with any of the other paths in the set.
\end{enumerate}

If $Z_i$ is an auxiliary variable, we will call its parent $Z'_i$. Conditions 2 and 3 follow using the same proof as given for non-AVs above. The first condition, however, requires more care.
The case of $V= Z'_i$ is permitted by assumption.

Suppose $V\ne Z'_i$. That means that the path $p_i[Z_i\sim V]$ goes through one of $Z'_i$'s incoming edges (and does not go through the auxiliary edges). 
This path exists in the graph $G_{E-}$. If $V$ is descendant of $y$ then $Z_i$ is a descendant of $y$, since the directed path $p_i$ does not get cut in $G_{E-}$. 
Similarly, suppose $(v\not \indep y|W_i)_{G_{E}}$, then using the Conditional Edge Lemma 2, $(v\not \indep y|W_i)_{G_{E-}}$. Since there is a directed, unblocked path from $v$ to $Z'_i$, 
$(Z'_i\not \indep y|W_i)_{G_{E-}}$, so using Theorem 1, $(Z_i\not \indep y|W_i)_{G_{E+}}$ - a contradiction. Therefore $(v \indep y|W_i)_{G_{E}}$, so $v$ satisfies (a).

\end{proof}

For the next proof, we will assume that the conditions in lemma \ref{lemma:auxnosharing} hold.

\begin{lemma}
For all $1\le i \le n$, there exists no unblocked path between $Z_i$ and $Y$, different from $p_i$, which includes edge $X_i \rightarrow Y$ and is composed only by edges from $p_1,...,p_i$.
\end{lemma}

\begin{proof}
Let $p$ be an unblocked path between $Z_i$ and $Y$, different from $p_i$, and assume that $p$ is composed only by edges from $p_1,...,p_i$. According to the ordering condition, 
if $Z_i$ or $Z'_i$ appears in some path $p_j$, with $j\ne i$, then $j>i$. Therefore, $p$ must start at $Z_i$, and take a non-auxiliary edge from $Z'_i$. Since $p$ is different from $p_i$,
it must contain at least one edge from $p_1,...,p_{i-1}$. Let $(v_1,V_2)$ denote the first edge in $p$ which does not belong to $p_i$. From lemma \label{lemma:auxnosharing},
it follows that $V_1$ must be a $z_k$ or $z'_k$ for some $k<i$, and the subpath $p_i[Z_i\sim V_1]$ and $(V_1,V_2)$ must point to $V_1$. This implies that $p$ is blocked by $V_1$ (collider), a contradiction.
\end{proof}

Using the same proof, we also get:

\begin{lemma}
For all $1\le i \le n$, there exists no unblocked path between $Z_i$ and some $W_i$, composed only of edges from $p_1,...,p_i$.
\end{lemma}

\begin{lemma}
For all $1\le i \le n$, there exists no unblocked path between $Z_i$ and $Y$, including edge $X_j \rightarrow Y$, with $j< i$, composed only of edges from $p_1,...,p_i$.
\end{lemma}

To finish the proof, we add a comment about auxiliary variables to Brito's Lemma 7:

\begin{lemma}
The coefficients of edges incident to $y$ are 0 unless they are part of the instrumental set.
\end{lemma}

\begin{proof}
Using CEL1, we can see that the coefficients are $\sigma_{zp_i.W}$. But these are the same in graph $G$ and $G_{E-}$ by CEL 2.
If the coefficient were non-zero in $G_{E-}$, then $\sigma_{zy.W}$ would be non-zero by d-separation (there is a directed edge from each $p_i$ to $y$), meaning that conditional independence would be violated.
\end{proof}

This completes the necessary proof modifications. We were able to sidestep issues of same-value structural parameters by ensuring that all intersections that might move across the auxiliary edges
happen with $i<j$, and are not relevant to the proof.

\section{Equivalence of IV Definitions}
\label{sec:equivIV}
For convenience, Definition \ref{def:iv} is restated here in its original (theorem) form:

\begin{theorem}
\label{thm:brito}
\citep{brito:pea02a} Given a linear model with graph $G$, the coefficients for a set of edges $E = \{(x_1, y), ..., (x_k, y)\}$ are identified if there exists triplets $(z_1, W_1, p_1), ..., (z_k, W_k, p_k)$ such that for $i=1, ..., k$,
\begin{enumerate}
\item $(z_i \indep y |W_i)_{G_{E-}}$, where $W_i$ does not contain any descendants of $y$ and $G_{E-}$ is the graph obtained by deleting the edges, $E$ from $G$,
\item $p_i$ is a path between $z_i$ and $x_i$ that is not blocked by $W_i$, and
\item if $1 \le i < j \le n$ the variable $z_j$ does not appear in path $p_i$; and, if paths $p_i$ and $p_j$ have a common
variable $V$, then both $p_i[V\sim Y]$ and $p_j[Z_j\sim V]$ point to $V$.
\end{enumerate}
If the above conditions are satisfied, we say that $Z$ is a \emph{generalized instrumental set} for $E$ or simply an \emph{instrumental set} for $E$.\footnote{Note that when $k=1$, $z_1$ is a conditional IV for $(x_1, y)$.  Further, if $z_1 = x_1$, then $x_1$ satisfies the single-door criterion for $(x_1 , y)$. The converse is not true, however (see appendix \ref{sec:cIVvsgIS}).}
\end{theorem}

We will show that the third condition in this theorem can be replaced with an assertion that the paths have no sided intersection. That is, the following theorem is equivalent:

\begin{theorem}
\label{thm:IS}
Given a linear model with graph $G$, the coefficients for a set of edges $E = \{(x_1, y), ..., (x_k, y)\}$ are identified if there exists triplets $(z_1, W_1, p_1), ..., (z_k, W_k, p_k)$ such that for $i=1, ..., k$,
\begin{enumerate}
\item $(z_i \indep y |W_i)_{G_{E-}}$, where $W$ does not contain any descendants of $y$ and $G_{E-}$ is the graph obtained by deleting the edges, $E$ from $G$,
\item $p_i$ is a path between $z_i$ and $x_i$ that is not blocked by $W_i$, and
\item the set of paths, $\{p_1, ..., p_k\}$ has no sided intersection.
\end{enumerate}
\end{theorem}

We will perform several reversible steps to show that whenever an instrumental set of one type exists, a set of the other must also exist.

\begin{lemma}
\label{lemma:halftrek}
There exist triples satisfying the conditions of theorem \ref{thm:IS}, iff there exist triples satisfying the theorem with the additional condition that
each $p_i$ can be decomposed into $p_i^0[z_i\sim a_1],p_i^1[a_1\sim a_2],...,p_i^m[a_m\sim y]$ and each $p_i^j$ is a half-trek.
\end{lemma}

\begin{proof}
$\Leftarrow$ follows directly, since any set of triples satisfying lemma \ref{lemma:halftrek} automatically satisfies theorem \ref{thm:IS}.

$\Rightarrow$ Suppose we have a set of triples satisfying theorem \ref{thm:IS}. Consider the $i$th triple $(z_i,W_i,p_i)$ from theorem \ref{thm:IS}. 
Decompose $p_i$ into $p_i^j...$, splitting at each collider $\rightarrow a_j \leftarrow$. Suppose $p_i^j$ is not a half-trek, and it is closest to $y$ (i.e. $p_i^h,h>j$ are all half-treks).
We define $z'_i$ as the last variable in $Left(p_i^j)$ from $a_i$ along $p_i^j$\footnote{Remember that since $p_i^j$ is an unblocked path from $a_j$ to $a_{j+1}$, 
it is a trek starting with one or more nodes in $Left$, and ending with nodes in $Right$}. We replace the triple $(z_i,W_i,p_i)$ with the new triple $(z_i',W_i,p'_i[z_i'\sim y])$. 

We now show that this new set of triples satisfies the definition of lemma \ref{lemma:halftrek}.

\begin{enumerate}
\item Suppose $(z'_i \notindep y |W_i)_{G_{E-}}$. This means that $z_i$ and $y$ are not d-separated given $W_i$, and as such there exists a path from $y$ to $z_i'$. 
But there is path from $z'_i$ to $z_i$ starting with $z'_i\rightarrow$, which is also unblocked by $W_i$ (reverse $p_i$). Combining those two paths gives a path between $y$ and $z_i$, meaning $(z_i \notindep y |W_i)_{G_{E-}}$,
a contradiction.
\item Since $p'_i$ is a subpath of $p_i[z'_i\sim x_i]$, it is a path between $z'_i$ and $x_i$ that is not blocked by $W_i$.
\item By construction, $p'_i$ can be decomposed into a set of half-treks $p_i^j$. Furthermore, since $\{p_1,p_2,..., p_k\}$ had no sided intersection, and $\{p'_1,p'_2,..., p'_k\}$
are subpaths of these original paths, and by the fact that $z'_i$ must have already been in $Left(p_i)$, we have $Right(p'_i) = Right(p_i)$, and $Left(p'_i) \subseteq Left(p_i)$ for all i.
Therefore, $\{p'_1,p'_2,..., p'_k\}$ must not have sided intersection, since if it did, $\{p_1,p_2,..., p_k\}$ would have also had this intersection.
\end{enumerate}
\end{proof}

\begin{corollary}
\label{corollary:nov}
If there exist triples satisfying lemma \ref{lemma:halftrek}, then the set of paths $\{p_1,...,p_k\}$ can only intersect at $z_1,...,z_k$, where $z_i$ is the instrumental variable.
\end{corollary}

\begin{proof}
If two paths have no sided intersection, then any node that is in both paths must be in $Right$ of one path, and in $Left$ of the other.
Since each path $p_i$ is composed of a set of half-treks, the only variable in $Left$ of $p_i$ is $z_i$, with the rest of the variables in $Right$ (all colliders $a_i$ are in both $Right$ and $Left$). Thus any intersection must happen at $z_i$, the instrumental variable.
\end{proof}

\begin{figure}
\centering
\includegraphics[width=.6\linewidth]{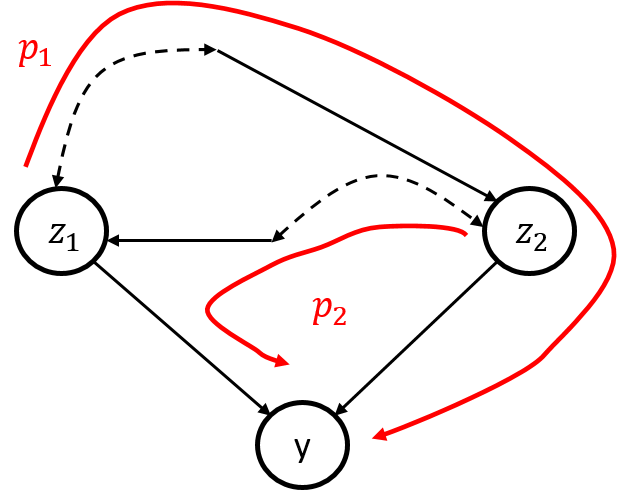}
\caption{The structure of intersecting paths in lemma \ref{lemma:switchorder}}
\label{fig:2intersect}
\end{figure}

\begin{lemma}
\label{lemma:switchorder}
There exist triples satisfying lemma \ref{lemma:halftrek} iff there exist triples satisfying the lemma AND $\forall z_i,z_j$, if $z_j$ is on path $p_i$, then $z_i$ is not on path $p_j$.
\end{lemma}

\begin{proof}
Using lemma \ref{lemma:halftrek}, we can generate a set of triples where all paths are composed of successive half-treks. 
Suppose that $\exists i,j$ s.t. $z_i$ is on path $p_j$ and $z_j$ is on path $p_i$.
We know that $p_i$ and $p_j$ must start with bidirected edges, making $z_i$ in Left in $p_i$ and $z_j$ in Left in $p_j$ (otherwise $z_i$ or $z_j$ would be in both Left and Right of its path, giving a sided intersection).
Similarly, we know that $p_i[... z_j \rightarrow ...]$ and $p_j[...z_i\rightarrow ...]$, since no sided intersection requires that the variables be in $Right$ of the other path.

We claim that there exist new triplets that no longer intersect each other: $(z_i,p_j[z_i\sim y],W_i')$ and $(z_j,p_i[z_j\sim y],W_j')$.
For example, in figure \ref{fig:2intersect}, $p_1$ and $p_2$ intersect at $z_1$ and $z_2$. We can modify the paths to be $z_1\rightarrow y$ and $z_2\rightarrow y$, which no longer intersect and form an equivalent valid set.

We now prove that in the switched triples, there exist $W_i'$ and $W_j'$ such that:

\begin{enumerate}
\item The modified paths have no sided intersection with each other, nor with other variables in the resulting instrumental set,
\item $W_i'$ and $W_j'$ do not block $p_j[z_i\sim y]$ and $p_i[z_j\sim y]$ respectively, and
\item $(z_i\indep y|W_i')_{G_{E-}}$ and $(z_j\indep y|W_j')_{G_{E-}}$.
\end{enumerate}

Notice that if these conditions are satisfied, the resulting set satisfies theorem \ref{thm:IS}.

We first show that there is no sided intersection. Note that $p'_i$ and $p'_j$ are sub-paths of the original $p_j$ and $p_i$, 
which by assumption have no sided intersection with any other paths in the instrumental set. The only modification now is that the paths
start at $z_i$ and $z_j$ respectively. No path intersects with $z_i$ or $z_j$ in the new triples, because originally $z_i$ and $z_j$ were the intersection
of two paths, one in $Left$ and one in $Right$, meaning that no other path could go through them - and now this intersection no longer exists, and all other variables are unchanged.
Thus the modified paths have no sided intersection with any other variable.

Finally, we show that there exist conditioning sets that satisfy the second and third requirements. We focus on $W'_i$, and $W'_j$ will hold by symmetry.

We divide into two possible cases: $W_j\cap Desc(z_i)_{G_{E-}}\ne \emptyset$ and $W_j\cap Desc(z_i)_{G_{E-}}= \emptyset$.

\begin{itemize}
\item $W_j\cap Desc(z_i)_{G_{E-}}\ne \emptyset$ - Note that $W_j$ does not block $p'_i$, since it doesn't block $p_j$.
Now, suppose for the sake of contradiction $(z_i\notindep y|W_j)_{G_{E-}}$. This means that $z_i$ is not d-separated from $y$ in $G_{E-}$, so there exists
an unblocked path from $y$ to $z_i$. But since $W_j$ conditions on a descendant of $z_i$, no matter how the path gets to $z_i$, it can cross a collider at $z_i$, and be extended by $p_j[z_j\sim z_i]$,
meaning that $(z_j\notindep y|W_j)_{G_{E-}}$, a contradiction. Finally, $W_j$ does not contain descendants of $y$. Therefore, we can use $W'_i = W_j$.
\item $W_j\cap Desc(z_i)_{G_{E-}}= \emptyset$ - In this case, we know that $y\notin Desc(z_i)_{G_{E-}}$, because if it were, we could create a path from $y$ to $z_j$ through $z_i$, 
by combining $p_j[z_j\sim z_i]$ with the directed path from $z_i$ to $y$ to show that $(z_j\notindep y|W_j)_{G_{E-}}$, a contradiction.

Consider $W'_i = W_i\setminus Desc(z_i)$. $W'_i$ does not block $p'_i$, since $p'_i$ is a directed path to the descendants of $z_i$ (there is no collider remaining). 
Finally, we need to show that $(z_i\indep y|W'_i)_{G_{E-}}$. Suppose not. This means that there exists a path $p_v$ from $y$ to $z_i$ which is not blocked by $W'_i$. 
We know that this path is blocked by $W_i$, so the blocking variable $v$ must be a descendant of $z_i$. 
Since the path starts at $y$, which is not a descendant of $z_i$ and goes to a descendant of $z_i$, it must come into the descendants of $z_i$ through an incoming edge.
This path must now get to $z_i$, but $W'_i$ has no conditioning in the descendants of $z_i$, so $p_v$ cannot cross a collider. The graph is acyclic, so $p_v$ cannot get to $z_i$ by following a directed path in $z_i$'s descendants. This is a contradiction. Therefore, $(z_i\indep y|W'_i)_{G_{E-}}$.
\end{itemize}

Since the conditions of theorem \ref{thm:IS} are satisfied for the new set, we can perform this procedure for all pairs of variables which intersect. 
The procedure will only need to be done at most once per pair of variables, since the resulting paths cannot increase the number of double-intersections.
The result is a set where $\forall z_i,z_j$, if $z_j$ is on path $p_i$, then $z_i$ is not on path $p_j$.
\end{proof}

Theorem \ref{thm:brito} requires a valid ordering of the variables. We showed that there are orderings of size 2, but in order 
to prove the theorem in general, we must show that there is a full ordering of all of the variables. To show this, we will first show that we can generate a set without intersection loops.

\begin{mydef}
An intersection loop is a sequence of paths $p_1,...,p_j$ where $\forall i$, $p_i$'s $Right$ intersects with $p_{i+1}$'s $Left$, and $p_j$'s $Right$ intersects with $p_1$'s $Left$.
\end{mydef}

\begin{figure}
\centering
\includegraphics[width=.6\linewidth]{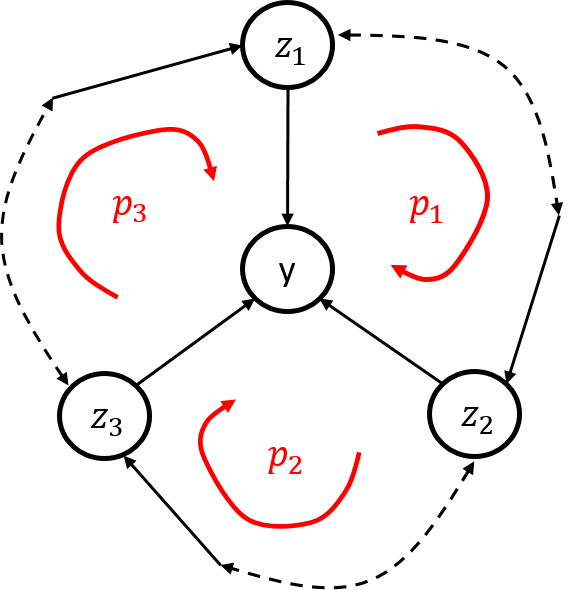}
\caption{An example of an intersection loop of size 3. Note that in this case there is no ordering of all 3 variables $i < j$ s.t. $z_j$ does not appear in path $p_i$.}
\label{fig:iloop3}
\end{figure}

An example of an intersection loop of size 3 is given in figure \ref{fig:iloop3}. Remember that each path can be decomposed into sets of half-treks WLOG,
so intersection loops are the only type of loop possible (corollary \ref{corollary:nov}).
Thankfully, the next lemma shows that any instrumental set can be modified such that there is no intersection loop.

\begin{lemma}
\label{lemma:iloop}
There exist triples satisfying lemma \ref{lemma:halftrek} iff there exist triples satisfying the lemma, AND there are no intersection loops between $\{p_1,...,p_k\}$.
\end{lemma}

\begin{proof}
We will generalize the proof of lemma \ref{lemma:switchorder} to work with an arbitrary amount of nodes. 
Using the same arguments as given in lemma \ref{lemma:switchorder}, the paths must all start with bidirected edges, and the only intersection allowed
is between the first element of each path, and the directed portion of other paths.

Suppose there is an intersection loop of size $n$, consisting of $p_1,...,p_n$, with corresponding triples $(z_1,p_1,W_1),(z_2,p_2,W_2),...,(z_n,p_n,W_n)$.
We claim that these triples can be replaced a new set: $(z_1,p_n[z_1 \sim y],W'_1),(z_2,p_1[z_2 \sim y],W'_2),...,(z_n,p_{n-1}[z_n \sim y],W'_n)$.

First note that each of the new paths is valid (since the original paths intersected from $Right$, meaning that $p_i[z_{i+1} \sim y]$ starts with $z_{i+1}\rightarrow$).
These new paths have no sided intersection (see lemma \ref{lemma:switchorder}). Furthermore, these new paths cannot be part of any intersection loop, since none of them start with bidirected edges.
This means that we only need to do one pass through all the loops in the original instrumental set to remove them all.

Finally, we mirror the arguments given in the proof of lemma \ref{lemma:switchorder} to show that there exist new conditioning sets for each $p_i$ that satisfy the conditions of lemma \ref{lemma:halftrek}.
Consider $W'_i$, for all $i=1...n$. We divide into two possible cases:

\begin{itemize}
\item $W_{i-1} \cap Desc(z_i)_{G_{E-}}\ne \emptyset$ - Using the same argument as in lemma \ref{lemma:switchorder}, $W'_i=W_{i-1}$ satisfies the requirements.
\item $W_{i-1} \cap Desc(z_i)_{G_{E-}}= \emptyset$ - Using the same argument as in lemma \ref{lemma:switchorder}, $W'_i=W_i\setminus Desc(z_i)$ satisfies the requirements.
\end{itemize}

Since the new set satisfies the requirements of lemma \ref{lemma:halftrek}, and the loop no longer exists, we can iteratively repeat the procedure for all intersection loops remaining
in the instrumental set, taking apart at most $\frac{k}{2}$ loops (if all paths are part of a loop of size 2). We are then left with a graph with no intersection loops.

\end{proof}

\begin{theorem}
There exists a set of triples satisfying theorem \ref{thm:brito} if and only if there exists a set of triples satisfying theorem \ref{thm:IS}.
\end{theorem}

\begin{proof}
$\Rightarrow$ The first two conditions are identical. The only difference is the third condition. Suppose we have a set satisfying theorem \ref{thm:brito}. If there is no intersection, then we have automatic satisfaction of lemma \ref{lemma:halftrek} and this theorem. If there is an intersection between two paths, then they share a variable
$V$, and both $p_i[V\sim Y]$ and $p_j[Z_j\sim V]$ point to $V$. 
Since $p_i[V\sim Y]$ points to $V$, $V\in  Left(p_i)$, and since the path is unblocked, it must point on to $z_i$, so $V \not\in Right(p_i)$.

Similarly, $p_j[Z_j\sim V]$ points to $V$, meaning that $V\in Right(p_j)$, and the path is unblocked, so it must go from $V$ to $x_j$, so $V\not\in Left(p_j)$. Therefore the two paths have no sided intersection.

$\Leftarrow$ The first two conditions are identical. We will focus on condition 3. Using lemma \ref{lemma:halftrek} and lemma \ref{lemma:iloop}, we can generate a set of triples which have no intersection except at the instrumental variables $z$. 
Since $z_i$ is in $Left$, any intersection must be in $Right$ of the intersecting path. 
This means that both $p_i[z_i \sim y]$ and $p_j[z_j\sim z_i]$ point to $z_i$, satisfying the second part of the third condition.

We generate an ordering for the variables by generating a directed intersection graph, where there is a directed arrow between $p_i$ and $p_j$ if $z_j$ appears in path $p_i$. 
Note that $z_j$ appears in path $p_i$ iff $p_j$'s Left intersects with $p_i$'s Right. By lemma \ref{lemma:iloop}, this graph is acyclic. We therefore can put the nodes in topological
order, giving us an ordering satisfying theorem \ref{thm:brito}.

\end{proof}

\section{Conditional IV vs Generalized Instrumental Set}
\label{sec:cIVvsgIS}

\begin{figure}[H]
\centering
\includegraphics[width=.6\linewidth]{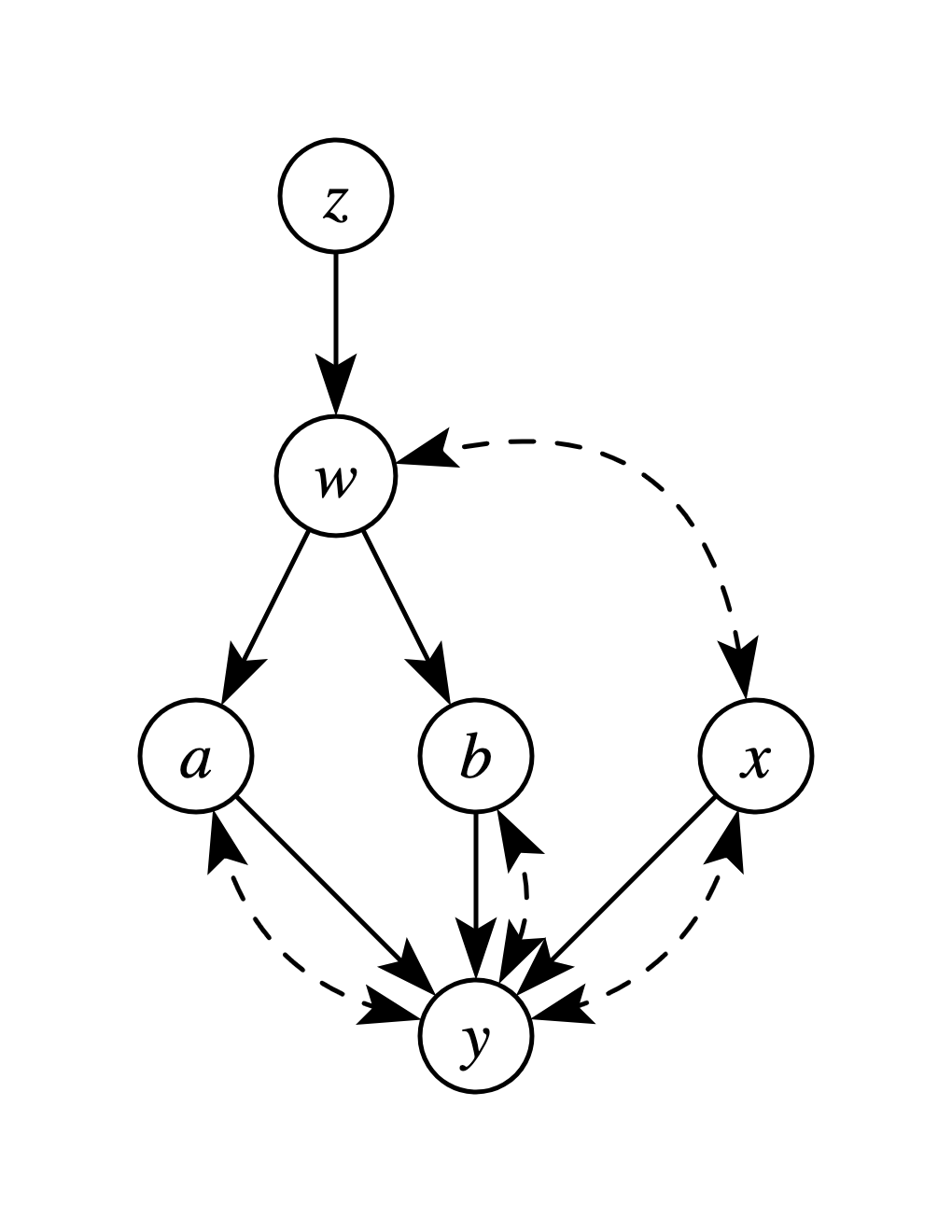}
\caption{A sample graph where $x\rightarrow y$ is not identifiable with generalized instrumental sets nor with quasi-AV sets, but is identifiable with a conditional IV. We can use $z|w$ as a single conditional instrumental variable.}
\label{fig:cIVex}
\end{figure}

While it might seem like generalized instrumental sets and quasi-AV sets are strictly more powerful than single conditional IVs, they have the requirement that the path from each instrument to its matched parent of $y$ contains no colliders (i.e., is unblocked given an empty conditioning set).

This requirement is not there for single conditional IVs. An example which takes advantage of this difference is given in figure \ref{fig:cIVex}.

\bibliography{book}  
\bibliographystyle{ims}

\end{document}